\documentclass[journal,draftclsnofoot,onecolumn,12pt,twoside]{IEEEtranTCOM}

\normalsize

\usepackage{amsmath, amsfonts, amssymb, amsthm}
\usepackage{bbm}
\usepackage{graphicx}
\usepackage{epstopdf}
\usepackage[noabbrev]{cleveref}
\usepackage{scalefnt}
\usepackage{indentfirst}
\usepackage{cite}
\usepackage{enumerate}
\usepackage{makecell}
\usepackage[caption=false,subrefformat=parens,labelformat=parens]{subfig}
\usepackage{scalerel}
\usepackage{color}
\usepackage{soul}
\definecolor{lightgray}{gray}{0.85}
\sethlcolor{lightgray}

\usepackage{algorithm}
\usepackage{algcompatible}
\usepackage{multirow}

\usepackage{pgfplots}
\pgfplotsset{compat=1.8}
\usepackage{tikz}
\usepgfplotslibrary{patchplots}
\usetikzlibrary{positioning, arrows, shapes}
\usetikzlibrary{shapes.geometric}
\usetikzlibrary{positioning}
\usetikzlibrary{patterns}
\usetikzlibrary{decorations}
\usetikzlibrary{spy}
\tikzset{
	font={\fontsize{10pt}{10}\selectfont}}

\algnewcommand\INPUT{\item[\textbf{Input:}]}
\algnewcommand\INITIAL{\item[\textbf{Initialization:}]}
\algnewcommand\OUTPUT{\item[\textbf{Output:}]}
\algnewcommand\RETURN{\item[\textbf{Return:}]}
\algnewcommand\ITER{\item[\textbf{Iteration:}]}
\algrenewcommand\algorithmiccomment[2][\small]{{#1\hfill\ #2}}

\theoremstyle{plain}
\newtheorem{theorem}{Theorem}
\newtheorem{lemma}{Lemma}
\newtheorem{proposition}{Proposition}

\theoremstyle{definition}

\theoremstyle{remark}

\usepackage{sansmathaccent}
\pdfmapfile{+sansmathaccent.map}
\usepackage{tikz}
\usepackage{pgfplots}
\tikzset{pic shift/.store in=\shiftcoord,
	pic shift = {(0,0)},
	pics/myfunc/.style args={scale #1 xshift #2 yshift #3}{
		code={	
			\begin{scope}[shift={\shiftcoord},scale=#1, every node/.append style={scale=#1}]

			\begin{axis}[
			,ymax=0.135
			,xmax=43
			,compat=newest
			,xshift=#2
			,yshift=#3			
			,axis lines=middle
			,every axis x label/.style={at={(current axis.right of 			origin)},anchor=west}
			,every axis y label/.style={at={(current axis.north west)},above=0mm}
			,samples=91, thick
			,domain=0.001:42
			,xtick={0.001,4,8,16,24,32,40}
			,ytick={0,0.0114,0.0272,0.125}
			,xticklabels={0, $\frac{N}{2m}$, $\frac{N}{m}$, $\frac{2N}{m}$,$\frac{3N}{m}$,$\frac{4N}{m}$,$\frac{5N}{m}$}
			,yticklabels={$0$, $\frac{1}{m\left| \sin\frac{\pi (2i_{\omega_{2}}+1)}{2m} \right|}$, $\frac{1}{m\left| \sin\frac{\pi (2i_{\omega_{1}}+1)}{2m} \right|}$, $1$}
			]
			\addplot+[no marks]
			{1/1024*abs(sin(180/pi*pi*128*x/1024)/sin(180/pi*pi*x/1024))};
			\draw[black, thin, dashed] (axis cs:0,0.0114) -- (axis cs:27.8, 0.0114);
			\draw[black, thin, dashed] (axis cs:0,0.0272) -- (axis cs:11.5, 0.0272);
			\draw[black, thin, dashed] (axis cs:4,0) -- (axis cs:4, 0.0796);
			\draw[red!70!black, thick] (axis cs:10,0) -- (axis cs:10, 0.0225);
			\draw[red!70!black, thick] (axis cs:30,0) -- (axis cs:30, 0.0075);

			\draw [blue, fill] (axis cs:27.8, 0.0114) circle (2pt);
			\draw [blue, fill] (axis cs:19.7, 0.0161) circle (2pt);
			\draw [blue, fill] (axis cs:11.5, 0.0272) circle (2pt);
			\draw [blue, fill] (axis cs:4, 0.0796) circle (2pt);
			\draw [red!70!black, fill] (axis cs:10, 0.0225) circle (2pt);
			\draw [red!70!black, fill] (axis cs:30, 0.0075) circle (2pt);

			\end{axis}
			
			\begin{scope}[xshift=#2, yshift=#3]
			\node[xshift=#2, yshift=#3] (c1) at (0,0) {};
			\node[left of=c1, xshift=0.8cm, yshift=-0.3cm] {0};
			\node[above of=c1, yshift=5cm] {$f(|\omega^{\ast} - \omega|)$};
			\node[right of=c1, xshift=6.7cm] {$|\omega^{\ast} - \omega|$};

			\node[xshift=#2, yshift=#3] at (2.0,-1.2) (c3) {$|\omega^\ast - \omega_{1}|$};
			\node[xshift=#2, yshift=#3] at (4,-1.2) {$|\omega^\ast - \omega_{2}|$};

			\draw[->, red, thick, xshift=#2, yshift=#3] (2.0,-0.95) .. controls (2.0,-0.5) and (1.55, -0.5) .. (1.55, -0.05);
			
			\draw[->, red, thick, xshift=#2, yshift=#3] (4.0,-0.95) .. controls (4.0,-0.5) and (4.75, -0.5) .. (4.75, -0.05);
			\end{scope}

			\end{scope}
			
			}
	}
}

\begin{document}
	\tikzstyle{block} = [draw, fill=black!20, rectangle, 
	minimum height=3em, minimum width=2em, text width=6.5em, text centered]
	\tikzstyle{sblock} = [draw, fill=black!20, rectangle, 
	minimum height=3em, minimum width=2em, text width=3.5em, text centered]
	\tikzstyle{smallblock} = [draw, fill=white, rectangle, minimum height=3em, minimum width=2em, text width=2.5em, text centered]
	\tikzstyle{whiteblock} = [draw=none, minimum height=3em, minimum width=5em, text width=6em, text centered]
	\tikzstyle{sum} = [draw, fill=blue!20, circle, node distance=1cm]
	\tikzstyle{input} = [coordinate]
	\tikzstyle{output} = [coordinate]
	\tikzstyle{pinstyle} = [pin edge={to-,thin,black}]
	\tikzstyle{branch}=[fill,shape=circle,minimum size=0pt,inner sep=0pt]


{\Large \textbf{Notice:} This work has been submitted to the IEEE for possible publication. Copyright may be transferred without notice, after which this version may no longer be accessible.}
\clearpage

\title{Channel Aware Sparse Transmission for Ultra Low-latency Communications in TDD Systems}

\author{\IEEEauthorblockN{Wonjun Kim, Hyoungju Ji, and Byonghyo Shim}

\IEEEauthorblockA{Institute of New Media and Communications and Department of Electrical and Computer Engineering, Seoul National University, Seoul, Korea\\
Email: \{wjkim, hyoungjuji, bshim\}@islab.snu.ac.kr}
\thanks{This work was supported by 'The Cross-Ministry Giga KOREA Project' grant funded by the Korea government(MSIT) (No. GK18P0500, Development of Ultra Low-Latency Radio Access Technologies for 5G URLLC Service).}
\thanks{Parts of this paper was presented at the VTC, Chicago, USA, August 27-30, 2018~\cite{wj_VTC} and ICC, Shanghai, China, May 20-24, 2019~\cite{wj_ICC}.}
}


\markboth{Submitted paper}%
{Submitted paper}

\maketitle

\vspace{-4em}
\begin{abstract}
Major goal of ultra reliable and low latency communication (URLLC) is to reduce the latency down to a millisecond (ms) level while ensuring reliability of the transmission. Since the current uplink transmission scheme requires a complicated handshaking procedure to initiate the transmission, to meet this stringent latency requirement is a challenge in wireless system design. In particular, in the time division duplexing (TDD) systems, supporting the URLLC is difficult since the mobile device has to wait until the transmit direction is switched to the uplink. In this paper, we propose a new approach to support a low latency access in TDD systems, called channel aware sparse transmission (CAST). Key idea of the proposed scheme is to encode a grant signal in a form of sparse vector. This together with the fact that the sensing mechanism preserves the energy of the sparse vector allows us to use the compressed sensing (CS) technique in CAST decoding. From the performance analysis and numerical evaluations, we demonstrate that the proposed CAST scheme achieves a significant reduction in access latency over the 4G LTE-TDD and 5G NR-TDD systems.
\end{abstract}

\begin{IEEEkeywords}
Ultra-reliable and low latency communications (URLLC), Time division duplexing (TDD), Compressed sensing
\end{IEEEkeywords}

%
\IEEEpeerreviewmaketitle

\section{Introduction}
%
%
%
%
\IEEEPARstart{F}{uture} mobile communication systems are expected to change our life by supporting wide variety of services and applications such as tactile internet, remote control, smart factories, and driverless vehicles, to name just a few~\cite{Machine}.
In order to support these diverse services and applications, new types of requirements other than the classical throughput requirement are needed~\cite{3GPP_1}.
One such requirement is the reduction of latency down to a millisecond level while ensuring reliability of the transmission~\cite{3GPP_2}.
To cope with this new requirement and related services, ITU introduced new use case called \textit{ultra-reliable and low latency communications} (URLLC)~\cite{ITU-R}.
Since it is not possible to satisfy the stringent latency requirement by a small makeshift of current 4G LTE systems, an entirely new uplink transmission scheme to support URLLC is required.


Recently, there have been some studies to achieve the latency reduction in the downlink transmission~\cite{Ji_URLLC,3GPP_3,System_design,Lee_packet}.
One simple approach is to transmit an urgent data without any reservations~\cite{Ji_URLLC}.
Also, an approach reserving resources in prior to the data scheduling has been proposed~\cite{3GPP_3}.
In~\cite{System_design}, an approach to dynamically multiplexing the enhanced mobile broadband (eMBB) and URLLC services has been proposed.
Also, a receiver technique to improve the reception quality and latency has been proposed in~\cite{Lee_packet}.

In the uplink direction, however, these approaches might not be applicable since the uplink transmission is subject to the complicated handshaking procedure with heavy signaling overhead.
Note that the signaling process requires a complicated interplay between the base station and mobile device, and thus it takes quite a bit of time for a mobile device to initiate the data transmission.
Indeed, it has been reported that the signaling for LTE scheduling takes more than 7ms even for the best scenario~\cite{sesia2011lte}.

In the future cellular systems, time division duplexing (TDD) system is expected to be a popular duplexing scheme due to the improved spectrum efficiency, better adaptation quality to asymmetric uplink/downlink traffics, low transceiver cost, and better support of the massive MIMO due to the channel reciprocity~\cite{Quek_1},\cite{Gao_1}.
In fact, since the main NR frequency band (e.g., the mid (3.3-3.8GHz) and high (24.25-29.5GHz) bands) is allocated as a TDD mode, supporting the URLLC in TDD system is of great importance~\cite{3GPP_3}.
However, satisfying the latency requirement in the TDD systems is far more difficult since the mobile device cannot transmit the data when the subframe is directed to the downlink (DL).
Thus, even though there is an urgent information to transmit, mobile device has no way but to wait until the transmit direction is switched to the uplink (UL).
For example, current 4G LTE TDD systems switch from DL to UL with half-frame-level (5ms) or frame-level (10ms) period so that the URLLC requirements cannot be satisfied with an ordinary processing~\cite{Quek_2},\cite{Ji_dynamic}.
One can naturally infer from this observation that a direct way to reduce the physical layer latency is to shorten the switching period up to the subframe-level (1ms) period or less.
Even in this case, it is not easy to support the short switching period in current 4G LTE systems due to the time-consuming and complicated handshaking process.

\begin{figure}[t]
\centering
\includegraphics[width=\columnwidth]{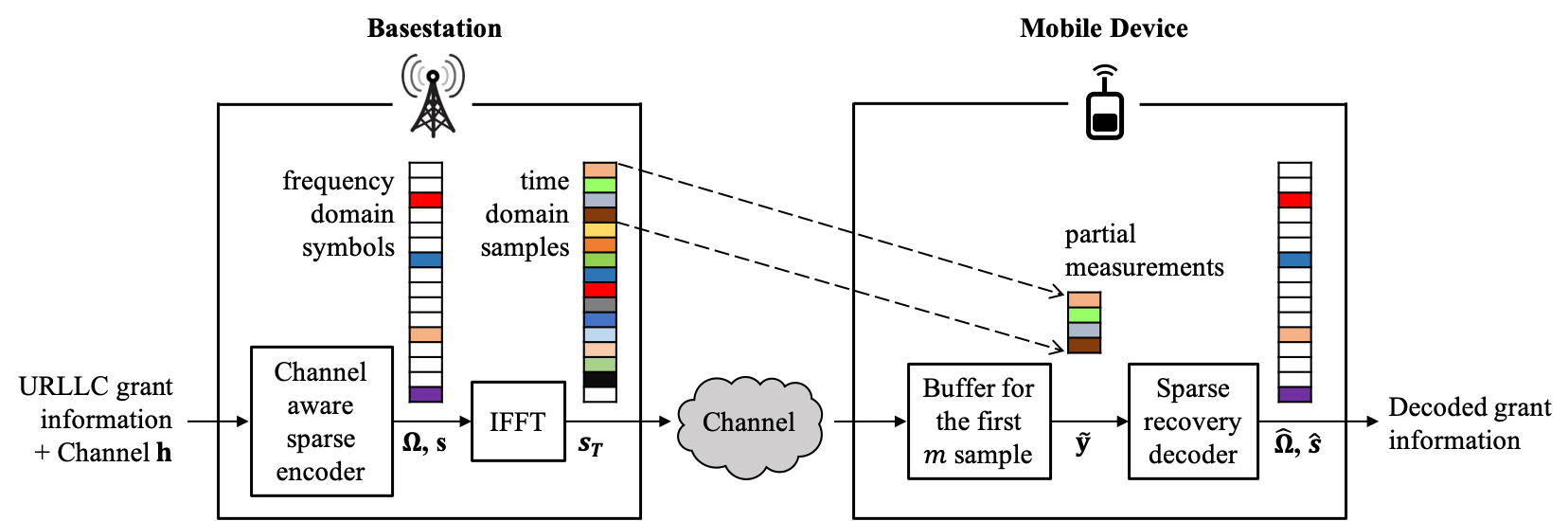}
\caption{Overall description of channel-aware sparse transmission (encoding and decoding) based on compressed sensing technique. The base station encodes the grant information (e.g., user ID, timing offset, and transmission band) into the small number of frequency-domain subcarriers (symbols). After receiving the early measurements $\tilde{\mathbf{y}}$, mobile device can decode the information using the sparse signal recovery algorithm.}
\label{fig:1}
\end{figure}

An aim of this paper is to propose a low latency uplink access scheme suitable for TDD-based URLLC systems.
Key feature of the proposed scheme is to transmit the latency sensitive information without waiting for the transmit direction change.
To be specific, the base station switches the transmit direction to UL right after sending the URLLC grant signal and hence a mobile device having the latency sensitive information can access the UL resources quickly.
To support the fast uplink access, we introduce a new grant signaling scheme, referred to as \textit{channel-aware sparse transmission} (CAST).
Key idea of CAST is to encode the URLLC grant information into a small number of subcarriers in the OFDM symbol.
In doing so, we make the frequency-domain OFDM symbol vector \textit{sparse} (see Fig.~\ref{fig:1}).
This together with the fact that the sensing matrix is a submatrix of the inverse discrete Fourier transform (IDFT) matrix allows us to use the compressed sensing (CS) principle in the decoding of the grant signal.
It is now well-known from the theory of CS that an accurate recovery of a sparse vector is guaranteed with a relatively small number of measurements as long as the sensing (measurement) process preserves the energy of an input sparse vector~\cite{Candes}.
In our context, this means that a mobile device can accurately decode the grant information with a small number of \textit{early arrived} received samples (see Fig.~\ref{fig:1}), which in turn means that UL access latency (latency of transmission and processing of the grant signal) can be reduced dramatically.

From the performance analysis in terms of the decoding success probability and also numerical evaluations on the latency sensitive data transmission, we demonstrate that the proposed CAST scheme is very effective and achieves fast uplink access.
In particular, in a realistic simulation setup, we observe that CAST achieves more than 80\% reduction in the uplink access latency over the 4G LTE and LTE-Advanced TDD systems.

The main contributions of this paper are as follows:
\begin{itemize}
\item We propose a low-latency signaling scheme based on the CS principle called CAST. In the proposed scheme, the base station encodes the grant information into a sparse vector and the mobile device decodes the packet using a sparse recovery algorithm. By using \textit{early arrived} samples in CAST decoding, we achieve a significant reduction in transmission and decoding latencies.
\item We develop the fast TDD access scheme based on CAST. To be specific, by mapping the user information to nonzero positions of a sparse vector derived from the environmental information, we can simplify the user identification process considerably.
\item We provide a performance analysis and empirical simulations to verify the reliability and latency gain of the proposed scheme. From these studies, we observe that the proposed CAST scheme achieves a significant reduction in access latency over the 4G LTE and 5G NR TDD systems.
\end{itemize}

The rest of this paper is organized as follows.
In Section II, we review the uplink access latency of the conventional TDD systems.
In Section III, we discuss the proposed CAST scheme and describe the encoding and decoding operations.
We also analyze the decoding success probability of the proposed CAST scheme.
In Section IV, we present simulation results to evaluate the performance and latency gains of CAST.
Finally, we conclude the paper in Section V.

We briefly summarize notations used in this paper.
We use uppercase boldface letters for matrices and lowercase boldface letters for vectors.
The operations $(\cdot)^{T}$ and $(\cdot)^{\ast}$ denote the transpose and conjugate transpose, respectively.
$\mathbb{C}$ and $\mathbb{R}$ denote the field of complex numbers and real numbers, respectively.
Also, $\mathbb{N}$ denotes the field of natural numbers.
$\lVert \cdot \rVert_{p}$ indicates the $p$-norm.
$\mathbf{I}_{N}$ is the $N \times N$ identity matrix.
$\mathbf{x}_{i}$ denotes the $i$-th column of the matrix $\mathbf{X}$ and $x_i$ is the $i$-th element of the vector $\mathbf{x}$.
$\mathbf{X}_{\Omega}$ is the submatrix of $\mathbf{X}$ that contains the columns as specified in the set $\Omega$ and $\mathbf{x}_{\Omega}$ is the vector constructed by picking the elements as specified in the set $\Omega$.

\section{Uplink Access Latency in TDD systems}
In this section, we briefly review the latency of TDD-based uplink transmission~\cite{3GPP_LTE}.
As mentioned, scheduling procedure is needed in 4G LTE systems to initiate the UL data transmission.
As illustrated in Fig.~\ref{fig:2}\footnote{In 4G LTE systems, the length of one radio frame is 10ms. Since one radio frame is divided into 10 subframes, the length of each subframe is 1ms. Also, each subframe consists of 14 OFDM symbols whose length is 66.7$\mu$s. Whereas, in the 5G New Radio (NR) systems, multiple numerologies are supported according to the various subcarrier spacing. In this paper, we consider the standard setting of 1ms subframe length with 15kHz subcarrier spacing.}, a mobile device sends a scheduling request (SR) signal to the base station when there is an information to transmit.
After receiving SR, the base station allocates resources and then sends the resource grant (RG) signal to the mobile device.
After receiving and decoding the RG signal, a mobile device begins to transmit the information to the base station in the assigned timing (resources).

\begin{figure}[t]
\centering
\includegraphics[width=0.95\columnwidth]{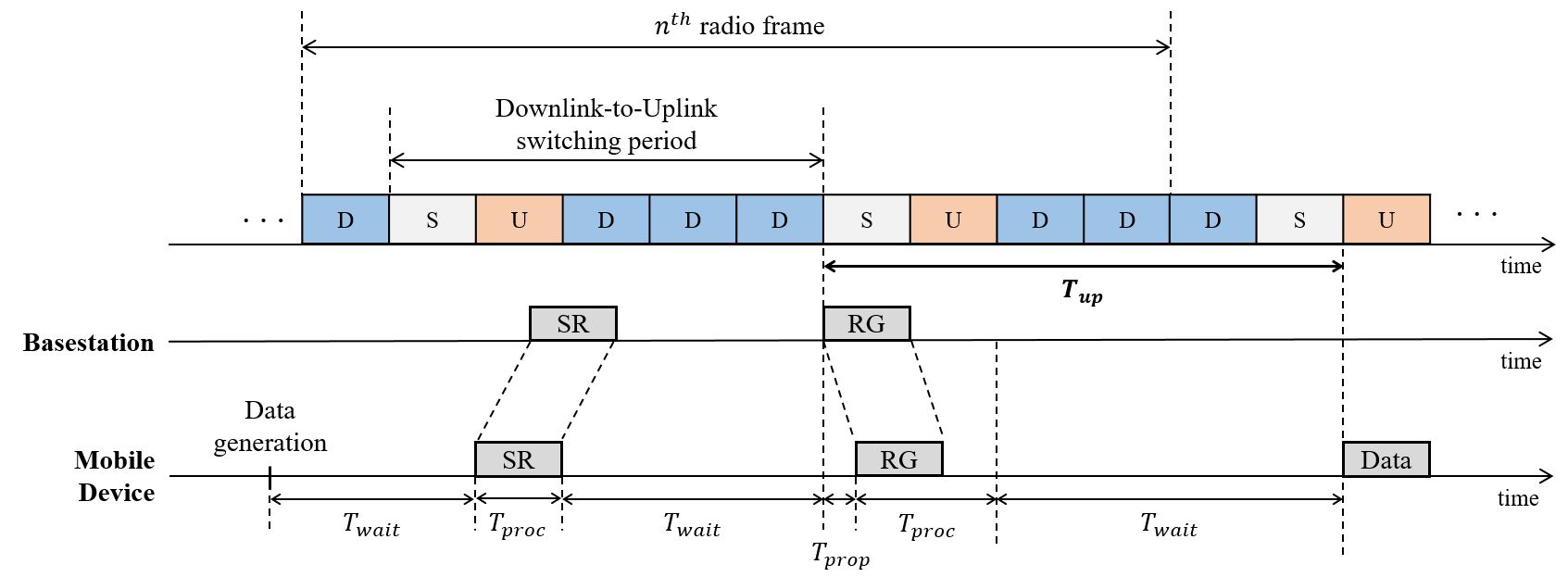}
\caption{An example of the scheduling-based uplink transmission in TDD systems. $\mathrm{D}$ and $\mathrm{U}$ denote the downlink subframe and uplink subframe, respectively. $\mathrm{S}$ is a special subframe required for switching the transmit direction. We assume that the uplink data is generated at the beginning of $n$-th radio frame.}
\label{fig:2}
\end{figure}

In the scheduling process, uplink access latency $T_{up}$, defined as the time duration from the transmission of the grant signal to the initiation of the data transmission,
can be expressed as the sum of three distinct latency components (see Fig.~\ref{fig:2}):
\begin{align}
T_{up} = T_{prop} + T_{proc} + T_{wait}.
\label{eq:latency}
\end{align}
\begin{itemize}
\item $T_{prop}$, called the propagation latency, is the time for a signal to travel from the base station to the mobile device
\item $T_{proc}$ is the processing latency for the grant signal
\item $T_{wait}$ is the waiting latency for the transmit direction change
\end{itemize}

Among these latency components, we put our emphasis on the reduction of the major components $T_{proc}$ and $T_{wait}$\footnote{The propagation latency $T_{prop}$ depends on the distance between the base station and mobile device. Hence, we consider it as a constant when the cell size is given.}.
First, $T_{proc}$ can be divided into two components: 1) the buffering latency $T_{buff}$ (the time to receive the grant signal) and 2) the decoding latency $T_{dec}$ (the time to decode the grant information).
For example, it takes around 1ms to buffer and decode the grant signal in the current 4G LTE systems~\cite{sesia2011lte}.
Clearly, this time would be too large to satisfy the URLLC latency requirement\footnote{In order to support URLLC services, 3rd Generation Partnership Project (3GPP) sets an aggressive requirement that a packet should be delivered with $10^{-5}$ packet error rate within 1ms transmission period~\cite{3GPP_2}.}.
$T_{wait}$ is caused by the periodic direction change in the TDD systems (see Fig.~\ref{fig:2}).
Since the current LTE TDD systems switch the transmit direction every 5ms or 10ms, a mobile device should wait until the direction is switched to UL to transmit the urgent data (even if the grant signaling is finished successfully).
Since this long switching period cannot satisfy the URLLC latency requirement, an access scheme with ultra short DL-to-UL switching period is needed for the success of URLLC.
When the switching period is short, one can notice that $T_{proc}$ would be a bottleneck to support fast UL access.
This is because a mobile device has enough time to decode the grant signal in the conventional TDD systems since the switching period (e.g., 5ms in LTE TDD systems) is much larger than $T_{proc}$.
However, when the switching period is very short (e.g, 1ms subframe-level switching), conventional grant signaling mechanism requiring all the received samples (e.g., 1024 samples in one OFDM symbol) to decode the grant information would not be a viable option due to the large $T_{proc}$ (e.g., 1ms in LTE systems).
In the following section, we describe the proposed CAST scheme to reduce $T_{proc}$ of the grant signal.

\section{Channel-aware Sparse Transmission}
\subsection{System Description of CAST}
Fig.~\ref{fig:1} depicts the block diagram of the proposed CAST scheme.
When designing the grant signal $\mathbf{s}$, the base station picks a small number, say $k$ out of $N$, of subcarriers.
For example, if the second and fifth subcarriers are chosen in the grant signal $\mathbf{s}$, then $\mathbf{s} = [0 \ s_{1} \ 0 \ 0 \ s_{2} \ 0 \ \cdots \ 0]$ ($s_1$ and $s_2$ are the symbols) and thus the support of $\mathbf{s}$ is $\Omega =\{2,5\}$.
In the CAST scheme, the granted (scheduled) user ID is encoded to the positions of the selected subcarriers\footnote{When the base station picks $k$ subcarriers out of $N$, then there are ${N \choose k}$ user IDs in total.
In the above example, $\Omega = \{2, 5\}$ is a user ID.} and the remaining grant information (e.g., uplink timing and transmission band) is encoded into the symbols.
We will say more about the encoding operation of CAST in Section III-B.


As mentioned, by using only small number of subcarriers, we make the grant signal vector $\mathbf{s}$ \textit{sparse}.
After the inverse fast Fourier transform (IFFT), the time-domain sample vector $\mathbf{s}_{t}=[s_{t}(1) \cdots s_{t}(N)]^{T}$ is transmitted through the fading channel.
The relationship between the transmit sparse grant signal $\mathbf{s}$ and the received time-domain samples $\mathbf{y}$ can be expressed as
\begin{align}
\mathbf{y}
&=
\mathbf{H}\mathbf{s}_{t} + \mathbf{v} \nonumber\\
&=
\mathbf{H}\mathbf{F}^{\ast}\mathbf{s} + \mathbf{v}
\label{eq:model}
\end{align}
where $\mathbf{H} \in \mathbb{C}^{N \times N}$ is the channel matrix, $\mathbf{F}^{\ast} \in \mathbb{C}^{N\times N}$ is the IDFT matrix, and $\mathbf{v} \sim \mathcal{CN}\left(0,\sigma_{v}^{2} \right)$ is the additive Gaussian noise vector.
Since the channel matrix $\mathbf{H}$ is the circulant matrix after removing the cyclic prefix, it can be eigen-decomposed by DFT matrix, i.e., $\mathbf{H} = \mathbf{F}^{\ast}\mathbf{\Lambda}\mathbf{F}$ where $\mathbf{\Lambda}$ is the diagonal matrix whose diagonal entry $\lambda_{ii}$ is the frequency-domain channel response for the $i$-th subcarrier.
Thus, we have

\begin{align}
\mathbf{y}
&=
(\mathbf{F}^{\ast}\mathbf{\Lambda}\mathbf{F})\mathbf{F}^{\ast}\mathbf{s} + \mathbf{v} \\
&=
\mathbf{F}^{\ast}\mathbf{\Lambda}\mathbf{s} + \mathbf{v} \\
&=
\mathbf{F}^{\ast}\mathbf{x} + \mathbf{v}
\label{eq:sparse_model}
\end{align}
where $\mathbf{x} = \mathbf{\Lambda}\mathbf{s}$.
It is worth mentioning that the supports of $\mathbf{s}$ and $\mathbf{x}$ are the same (i.e., nonzero positions of $\mathbf{s}$ and $\mathbf{x}$ are the same).

In the context of CS, $\mathbf{x}$ and $\mathbf{F}^{\ast}$ serve as the input vector and sensing matrix, respectively.
Since $\mathbf{F}^{\ast}$ preserves the signal energy of $\mathbf{x}$, by using properly chosen sparse recovery algorithm, the sparse vector $\mathbf{x}$ can be readily recovered from $\mathbf{y}$ with a small number of measurements.
Interestingly, this means that we only need a small number of \textit{early arrived} samples in $\mathbf{y}$ to decode the grant informations.
The corresponding partial measurement vector $\tilde{\mathbf{y}} \in \mathbb{C}^{m \times 1} (m \ll N)$ constructed from early arrived samples can be expressed as
\begin{align}
\tilde{\mathbf{y}}
&=
\mathbf{\Pi}\mathbf{y} \\
&=
\mathbf{\Pi}\mathbf{F}^{\ast}\mathbf{x} + \tilde{\mathbf{v}} \\
&=
\mathbf{A}\mathbf{x} + \tilde{\mathbf{v}}
\label{eq:model_1}
\end{align}
where $\mathbf{\Pi} = [ \mathbf{I}_{m} \ \mathbf{0}_{m \times (N-m)} ]$ is the matrix  to select the first $m$ samples among $N$ time-domain samples, $\tilde{\mathbf{v}} = \mathbf{\Pi}\mathbf{v}$ is the modified noise vector, and $\mathbf{A} = \mathbf{\Pi}\mathbf{F}^{\ast}$ is the partial IDFT matrix consisting of the first $m$ consecutive rows of $\mathbf{F}^{\ast}$.

As mentioned, the grant information is conveyed from both subcarrier indices and symbols and thus the decoding process is divided into two steps: 1) support identification to find out the nonzero positions of $\mathbf{s}$ vector and 2) symbol detection in nonzero positions.
First, for the decoding of the granted user ID, a mobile device needs to identify the support of $\mathbf{x}$, which is done by the sparse recovery algorithm~\cite{GOMP},\cite{Wei}.
After identifying the support $\Omega$, a mobile device decodes the remaining grant information by detecting the symbol vector $\hat{\mathbf{s}}_{\Omega}$.
Note that, after removing the components associated with the non-support elements in \eqref{eq:model_1}, the system model can be converted into the overdetermined system model ($m > k$).
For example, if $\Omega = \{2, 5\}$, then the system model in \eqref{eq:model_1} is simplified to $\tilde{\mathbf{y}} = \begin{array}{cc} [\mathbf{a}_{2} & \mathbf{a}_5] \end{array} \bigg[\begin{array}{c} x_{2} \\ x_{5} \end{array}\bigg] + \tilde{\mathbf{v}}$.
In detecting symbols $x_{2}$ and $x_{5}$, conventional technique such as the linear minimum mean square error (LMMSE) estimator followed by the symbol slicer can be used.

The benefits of CAST can be summarized as follows.
First and foremost, support identification for the decoding of the grant signal $\mathbf{s}$ is done with a small number of time-domain samples.
When compared to the conventional signaling mechanism in which all received samples are needed to decode the grant information, buffering latency $T_{buff}$ can be reduced by the factor of $m/N$.
For example, if $m=128$ and $N=1024$, then $T_{buff}$ would be reduced by the factor of $1/8$\footnote{Based on the principle of CS, an accurate recovery of the sparse vector is possible as long as $m \geq ck\log N$ where $c$ is the scaling constant ($c \approx 4$ as a ballpark number~\cite{Candes}).
For instance, when $N=1024$ and $k=3$, one can readily apply CS technique with $m \approx 120$ measurements.}.
Second, a channel information is unnecessary in the support identification process.
Recall that the sensing matrix $\mathbf{A}$ in \eqref{eq:model_1} is constructed only by the submatrix of IDFT matrix and what we need to do is to find out the nonzero positions of $\mathbf{x}=\mathbf{\Lambda}\mathbf{s}$, not the actual values.
Thus, we do not need the channel information in the support identification process.
Third, the implementation cost and the computational complexity of CAST is very low.
In particular, since the sparsity $k$ is small\footnote{The size of grant information excluding the user ID would be tiny for most of URLLC scenarios \cite{Ji_URLLC}. Hence, the small number $k$ of subcarriers is enough to convey the information. For example, when packet consists of 16 bits for grant information and 64 bits for user ID (RNTI), then we can use $N=1024$ and $k=8$ subcarriers with the QPSK modulation.} and also known to the mobile device, one can decode the grant information using a simple sparse recovery algorithm such as orthogonal matching pursuit (OMP)~\cite{Tropp}.
We will show in the next subsections that by choosing nonzero positions deliberately, support identification can be finished in just two iterations.

\subsection{Encoding Operation in CAST}
Since the decoding of the grant signal is done by the support identification, accurate identification of the support is of great importance for the success of CAST.
In general, when the system matrix is generated at random, the support identification performance would not be affected by the choice of support.
In the CAST scheme, however, the system matrix is constructed from IDFT matrix and the sparse vector $\mathbf{x} = \mathbf{\Lambda}\mathbf{s}$ is the product of the frequency-domain channel $\mathbf{\Lambda}$ and the sparse grant signal $\mathbf{s}$ so that both \textit{system matrix} and \textit{channel state} affect the decoding performance.

First, support identification performance depends heavily on the channel state.
For example, if a selected subcarrier $s_{i}$ undergoes a deep fading in the frequency-selective channel (i.e., $\lambda_{ii} \approx 0$), then an accurate identification of the nonzero position $x_{i} = \lambda_{ii}s_{i}$ would not be possible.
Since the DL channel information can be derived from the UL channel estimation via the channel reciprocity in TDD systems~\cite{Gao_1}, it would be desirable to choose indices of subcarriers having the highest subchannel gains as support elements (i.e., $\Omega = \arg\underset{\left\vert \Omega \right\vert = k}{\max}\lVert \mathbf{h}_{\Omega} \rVert_{2}$).
In doing so, one can reduce the chance of the decoding failure significantly.

\begin{figure}[t]
\centering
\includegraphics[width=0.55\linewidth]{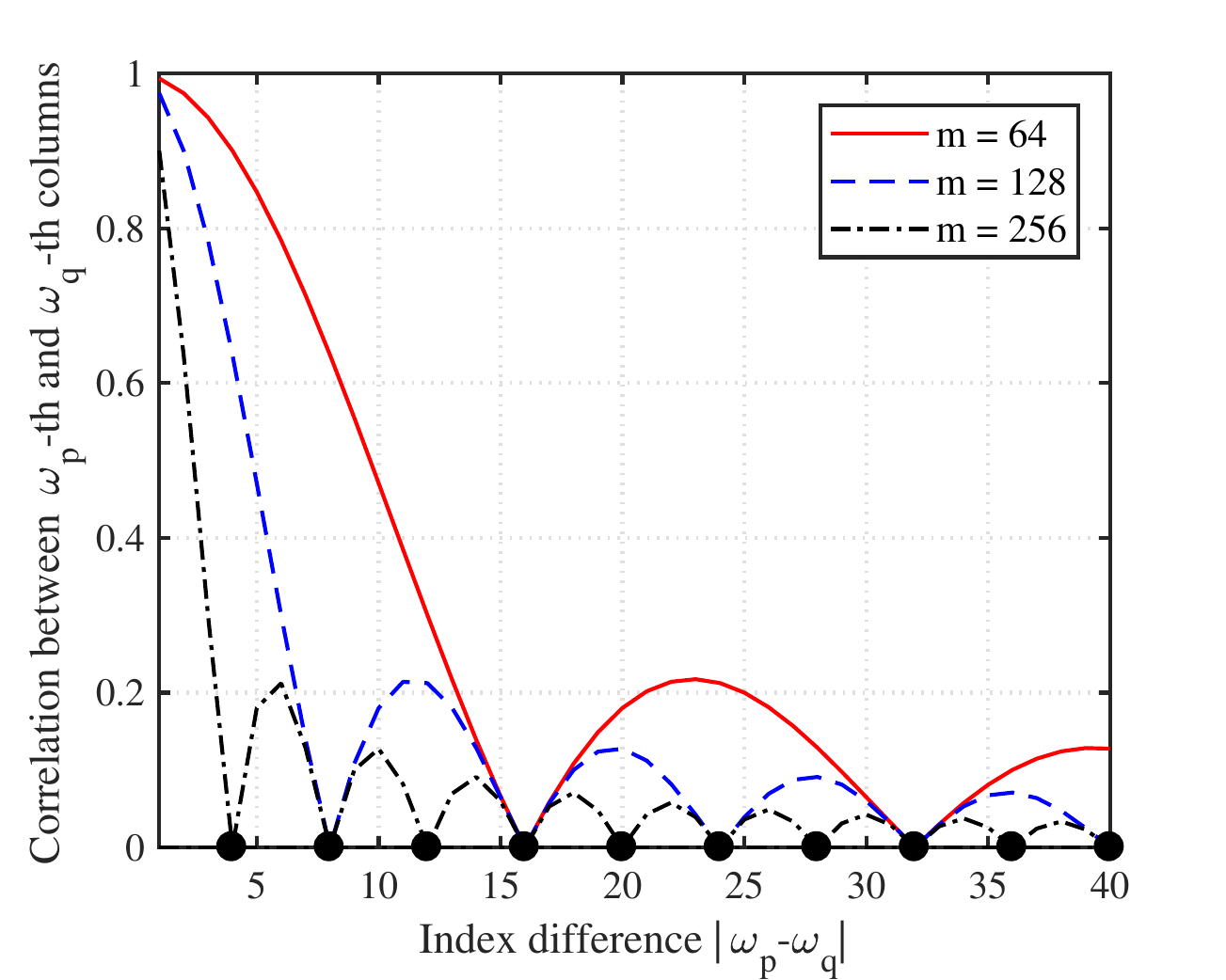}
\caption{Column correlation between $\mathbf{a}_{\omega_p}$ and $\mathbf{a}_{\omega_q}$  as a function of index difference $\left\vert\omega_p-\omega_q\right\vert$ ($N=1024$).}
\label{fig:column_correlation}
\end{figure}

\begin{figure}[t]
\centering
\includegraphics[width=0.9\linewidth]{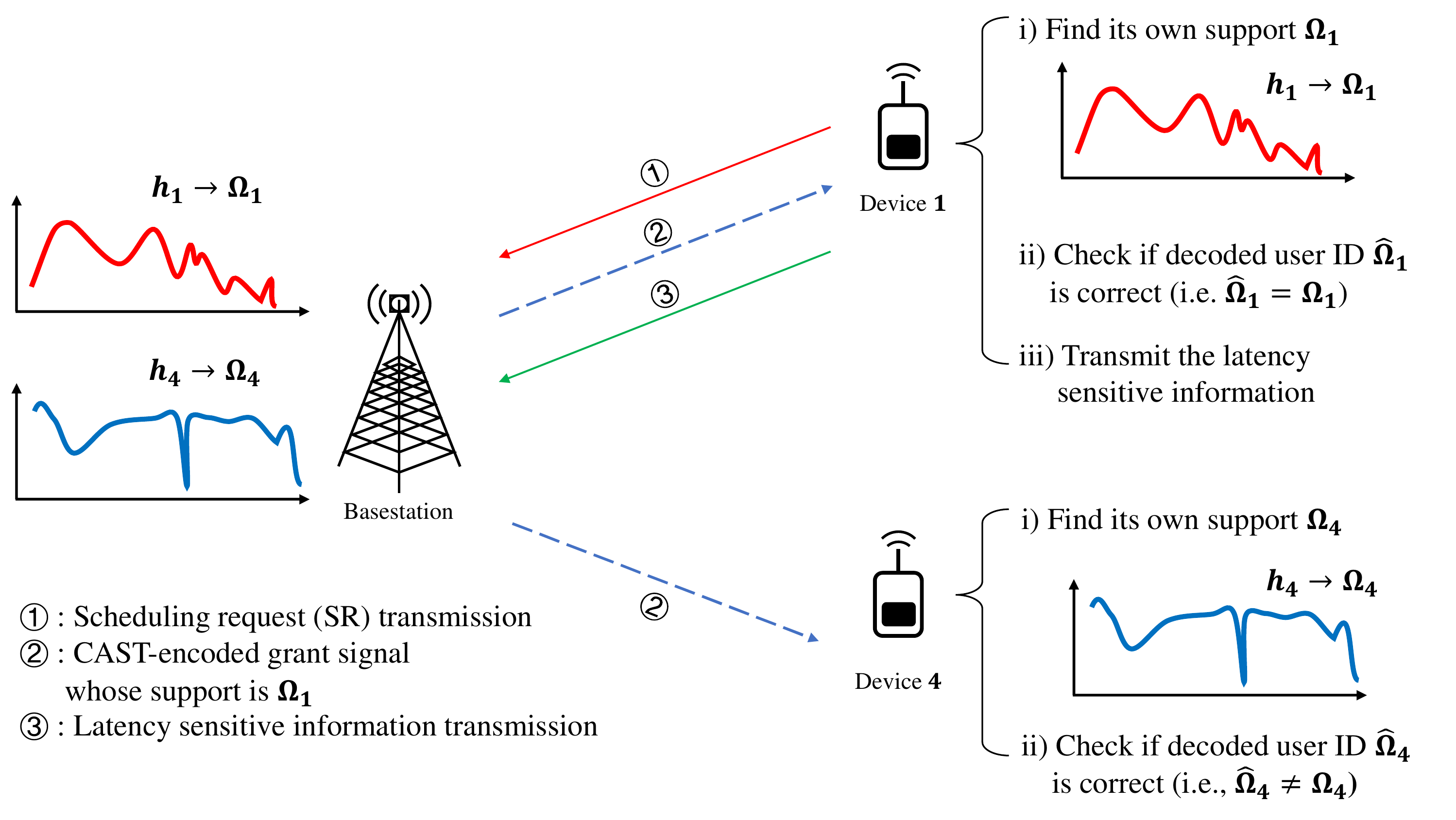}
\caption{Illustration of the CAST-based access in the TDD systems.}
\label{fig:CAST_access}
\end{figure}

Second, the support identification performance depends also on the correlation between columns in the system matrix $\mathbf{A}$.
In many greedy sparse recovery algorithms, such as OMP, an index of a column in $\mathbf{A}$ that is maximally correlated to the partial measurement $\tilde{\mathbf{y}}$ is chosen as an estimate of the support element~\cite{Tropp}.
Therefore, if two columns of $\mathbf{A}$ are strongly correlated and only one of these is associated with the nonzero values in $\mathbf{x}$, then it might not be easy to distinguish the right column (column associated with the nonzero value) from wrong one in the presence of noise.
Fortunately, since all entries of $\mathbf{A}=\mathbf{\Pi}\mathbf{F}^{\ast}$ are known in advance, we can alleviate this event by considering the column correlation of $\mathbf{A}$ in the support selection.
Specifically, let $f(\omega_{p}, \omega_{q})$ be the correlation between $\omega_{p}$ and $\omega_{q}$-th columns in $\mathbf{A}$, then we have
\begin{align}
f(\omega_{p}, \omega_{q})
=
\frac{1}{m}\left\vert\sum_{l=1}^{m}e^{-j2\pi(\omega_{p}-1)(l-1)/N} e^{j2\pi(\omega_{q}-1)(l-1)/N}\right\vert 
=
\label{eq:column_correlation}
\frac{1}{m}\left| \frac{\sin\frac{\pi m(\omega_{p}-\omega_{q})}{N}}{\sin\frac{\pi (\omega_{p}-\omega_{q})}{N}} \right|.
\end{align}
Since $f(\omega_{p}, \omega_{q})$ depends only on the absolute difference between $\omega_{p}$ and $\omega_{q}$, we will henceforth denote it as $f(\left\vert \omega_{p} - \omega_{q} \right\vert)$.
One can easily see that columns $\mathbf{a}_{\omega_{p}}$ and $\mathbf{a}_{\omega_{q}}$ are (near) orthogonal (i.e., $f(\left\vert \omega_{p} - \omega_{q} \right\vert) \approx 0$) if $\left\vert \omega_{p} - \omega_{q} \right\vert \approx c\frac{N}{m}$ for some integer $c$ (see Fig.~\ref{fig:column_correlation}).
Thus, by choosing the subcarrier indices from the set of the orthogonal columns in $\mathbf{A}$, accuracy of the support identification can be improved significantly.

\begin{algorithm}[t]
\renewcommand{\thealgorithm}{}
\floatname{algorithm}{}
	\caption{\textbf{Algorithm 1} The proposed CAST-based access}
	\begin{algorithmic}[1]
	\item[\textbf{Input:}]
	$\mathbf{h}\in\mathbb{C}^{N}$, $\mathbf{A}\in\mathbb{C}^{m\times N}$, $k\in\mathbb{N}$, $\Sigma = \{1, \cdots, N\}$
	\STATE Mobile device finds its own support $\Omega$ and base station selects support of the granted user via the following 3 steps
	\STATE\hspace{\algorithmicindent} $\omega^{\ast} = \arg\underset{\omega\in\Sigma}{\max} \lVert\mathbf{h}_{\omega}\rVert_{2}$
	\COMMENT{[Select index corresponding to the maximal channel gain]}
	\STATE\hspace{\algorithmicindent} $\Gamma = \{\gamma \in \Sigma\ \vert\ f(\gamma,\omega^{\ast}) \approx 0 \} \cup \{\omega^{\ast}\}$
	\COMMENT{[Design the index set of (near) orthogonal columns]}
	\STATE\hspace{\algorithmicindent} $\Omega = \arg\underset{\left\vert\Omega\right\vert=k, \Omega\subseteq\Gamma}{\max} \lVert\mathbf{h}_{\Omega}\rVert_{2}$
	\COMMENT{[Determine $\Omega$ corresponding to the $k$ largest channel gains]}
	\STATE Base station transmits the CAST-based grant signal $\mathbf{s}$ using $\Omega$
	\STATE Using a small number of early arrived samples, the mobile device decodes the CAST signals
	\STATE After the decoding, a mobile device sends the latency sensitive information immediately
	\end{algorithmic}
	\label{alg:selection}
\end{algorithm}

In summary, the support selection rule considering the channel state and system matrix is given by
\begin{align}
\Omega = \arg\underset{\left\vert\Omega\right\vert=k, \Omega\subseteq\Gamma}{\max} \lVert\mathbf{h}_{\Omega}\rVert_{2}
\label{eq:support_selection}
\end{align}
where $\Gamma$ is the index set of the orthogonal columns.
Overall grant procedure can be summarized as follows.
First, each and every mobile device finds its own support $\Omega$ (user ID) using \eqref{eq:support_selection}.
Exploiting the channel reciprocity, the base station can also figure out the user IDs of all mobile devices using \eqref{eq:support_selection}.
Second, after receiving SR, the base station transmits the CAST-based grant signal to the desired mobile device.
Using a small number of early arrived received samples, the mobile device can decode the grant signal.
Specifically, if the decoded support $\hat{\Omega}$ is equivalent to its own support $\Omega$ (i.e., $\hat{\Omega} = \Omega$), the grant signal is decoded successfully and thus the mobile device sends the (latency sensitive) information immediately (see Fig.~\ref{fig:CAST_access}).
The proposed CAST-based access procedure is summarized in Algorithm 1.

\subsection{Decoding Process in CAST}
\subsubsection{Basic Decoding}
As mentioned, key operation of the CAST decoding is to find out the support $\Omega$.
In other words, main task of decoding is to find $k$ nonzero positions of $\mathbf{x}$ vector from the received vector $\tilde{\mathbf{y}} = \mathbf{A}\mathbf{x} + \tilde{\mathbf{v}}$.
Note that this setup is common in many CS studies~\cite{CS_Trick}.
In our case, by exploiting the orthogonality of the columns associated with nonzero positions of $\mathbf{x}$, we can further simplify the support identification process.

To be specific, in the first iteration, a column maximally correlated with $\tilde{\mathbf{y}}$ is chosen as an estimate of support element $\hat{\omega}_{i}$.
Since columns associated with the support $\Omega$ are chosen from the set of orthogonal columns, remaining columns should be orthogonal to the column chosen in the first iteration.
In the second iteration, therefore, we choose $k-1$ best columns among those orthogonal to the firstly chosen column.
Thus, in contrast to the conventional greedy sparse algorithm in which $k$ iterations are required, the proposed CAST decoding is finished with only two iterations.
After this, a mobile device checks whether it is granted or not by comparing the decoded support $\hat{\Omega}$ and its own support $\Omega$.
If $\hat{\Omega} = \Omega$, remaining grant information is obtained by decoding the symbols associated with the support position.
\begin{figure}[t]
\centering
\subfloat[exact support identification]{\includegraphics[width=0.43\linewidth]{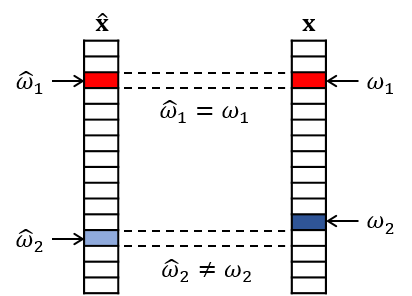}}
\hspace{1em}
\subfloat[$\tau$-close support identification]{\includegraphics[width=0.53\linewidth]{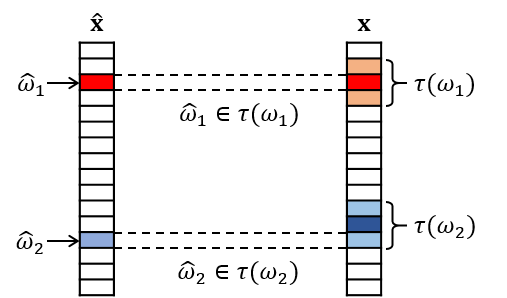}}
\caption{When $k=2$, $\Omega=\{\omega_{1}, \omega_{2}\}$, $\hat{\Omega} = \{\hat{\omega}_{1}, \hat{\omega}_{2}\}$, and $\tau=2$, success decisions for the exact support identification and $\tau$-close support identification are described : (a) The support identification is failed since $\hat{\omega}_{2} \neq \omega_{2}$. (b) The support identification is successful since $\hat{\omega}_{1} \in \{\omega_{1}-1, \omega_{1}, \omega_{1}+1\}$ and $\hat{\omega}_{2} \in \{\omega_{2}-1, \omega_{2}, \omega_{2}+1\}$.}
\label{fig:tau_close_detection_model}
\end{figure}

\subsubsection{$\tau$-close Support Identification}
Since the correlation between the adjacent columns in $\mathbf{A}=\mathbf{\Pi}\mathbf{F}^{\ast}$ is large (see \eqref{eq:column_correlation}), a column adjacent to the correct one might be chosen as a support element by mistake.
To avoid this type of mistake, we propose an improved scheme relaxing the success condition in the support identification.
Basic idea of the proposed strategy, called $\tau$-close support identification, is to regard the selected index as the correct one if the selected position is close to the true one.
That is, a chosen index $\hat{\omega}_{i}$ is considered as the correct one if it is not too far away from the true index $\omega_{i} \in \Omega$, i.e., $\hat{\omega}_{i} \in \{ \omega_{i}-\tau+1, \cdots, \omega_{i}, \cdots, \omega_{i}+\tau-1 \}$ (see Fig.~\ref{fig:tau_close_detection_model})\footnote{In a practical scenario, due to the channel variation or mismatch in the transmitter and receiver circuitry, the channel reciprocity might not be perfect. Due to this reason, the true support chosen by the mobile device might be slightly different from that chosen by the base station. By using the $\tau$-close support identification, this type of decoding failure can be also prevented.}.
In fact, as long as $\tau$ is smaller than the half of the minimum distance between any two orthogonal columns, a chosen index $\hat{\omega}_{i}$ can be replaced by $\omega_{i}$ and thus the decoding error can be prevented.
Since $\mathbf{x}$ is the sparse vector and hence the number of nonzero elements is small, as long as the difference between $\hat{\omega}_{i}$ and $\omega_{i}$ is small, there would not be any confusion caused by the $\tau$-close support identification.
\begin{figure}[t]
\centering
\includegraphics[width=0.55\linewidth]{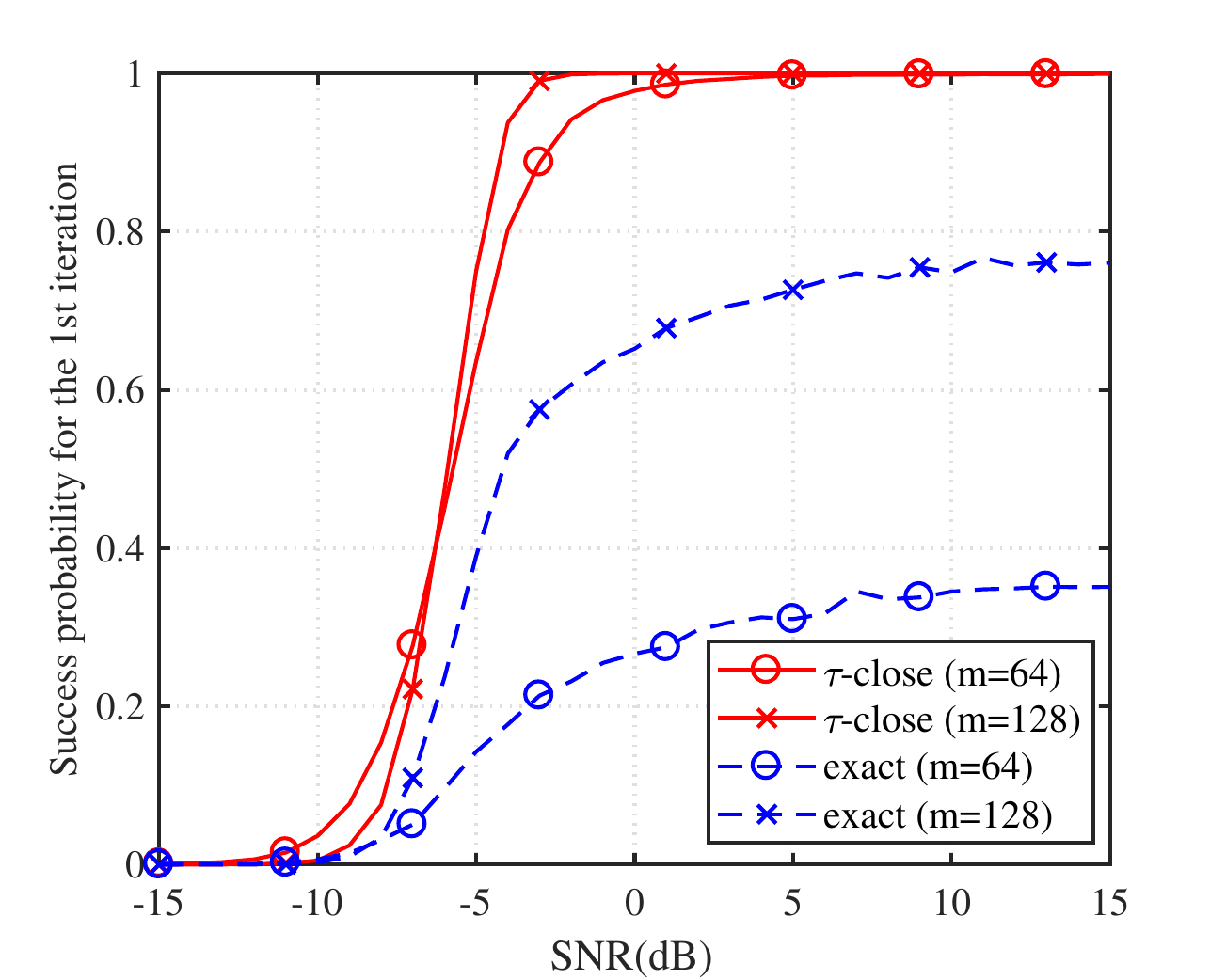}
\caption{Comparison between $\tau$-close support identification and conventional (exact) support identification in the first iteration using $\tau = \frac{N}{2m}$ ($N=1024$ and $k=12$)}
\label{fig:tau_close}
\end{figure}
\begin{algorithm}[t]
\renewcommand{\thealgorithm}{}
\floatname{algorithm}{}
	\caption{\textbf{Algorithm 2} The proposed CAST decoding algorithm}
	\begin{algorithmic}[1]
	\INPUT 
	$\tilde{\mathbf{y}}\in\mathbb{C}^{m}$, $\mathbf{A}\in\mathbb{C}^{m\times N}$, $k\in\mathbb{N}$, $\tau \in \mathbb{N}$, $\mathbf{h}\in\mathbb{C}^{N}$
	\STATE $\hat{\omega}_{1} = \arg\underset{\omega}{\max} \lVert\mathbf{a}^{\ast}_{\omega}\tilde{\mathbf{y}}\rVert_{2}$
	\STATE $\Gamma = \{\gamma\in\Sigma\ \vert\ f(\gamma, \hat{\omega}_{1}) \approx 0 \}$
   	\STATE (\textbf{Identification}) Select indices $\{\hat{\omega}_{t}\}_{t=2,\ldots,k}$ corresponding to $k-1$ largest entries in $\mathbf{A}_{\Gamma}^{\ast}\tilde{\mathbf{y}}$
   	\STATE $\hat{\Omega} = \{\hat{\omega}_{1}, \hat{\omega}_{2}, \cdots ,\hat{\omega}_{k}\}$
   	\STATE (\textbf{$\tau$-close support identification}) Check $\left\vert\hat{\omega}_{i} - \omega_{i} \right\vert < \tau$ for $i\in\{1,\cdots,k\}$ 
   	\IF{$\tau$-close support identification is successful}
   		\STATE (\textbf{Estimation of} $\tilde{\mathbf{s}}_{\Omega}$) $\tilde{\mathbf{s}}_{\Omega} = \arg\underset{\mathbf{u}}{\max}\lVert \tilde{\mathbf{y}} - \mathbf{A}_{\Omega}\mathbf{\Lambda}_{\Omega}\mathbf{u} \rVert_{2}$
   		\STATE (\textbf{Symbol slicing}) $\hat{\mathbf{s}}_{\Omega} = Q(\tilde{\mathbf{s}}_{\Omega})$
   	\ENDIF
	\OUTPUT $\hat{\Omega}, \hat{\mathbf{s}}_{\Omega}$
	\end{algorithmic}
	\label{alg:decoding}
\end{algorithm}

In Fig.~\ref{fig:tau_close}, we plot the success probability for the first iteration.
As discussed, since $k-1$ columns chosen in the second iteration are orthogonal to the column chosen in the first iteration, successful decoding in the first iteration is crucial for the success of the overall CAST decoding.
In our simulations, we compare the CAST decoding performance with and without the $\tau$-close support identification.
We observe that the $\tau$-close support identification is very effective and outperforms the conventional support identification by a large margin, which will be translated into the gain in decoding performance.
For example, when $m = 128$, the $\tau$-close support identification is perfect in most of SNR regimes under test but the conventional support identification performs poor and cannot be better than 0.8.
In Algorithm 2, we summarize a refined CAST decoding algorithm incorporating the $\tau$-close support identification.

\subsection{CAST Performance Analysis}
In this subsection, we present the success probability of user identification in the proposed CAST scheme.
By the successful user identification, we mean that all the true support elements are chosen by the CAST decoding process (i.e., $\hat{\Omega} = \Omega$).
As mentioned, one support element is chosen in the first iteration and the remaining $k-1$ support elements are chosen in the second iteration.
Thus, the success probability of user identification is expressed as
\begin{align}
\mathrm{P}_{succ}
&=
\mathrm{P} ({\hat{\Omega} = \Omega}) \\
&=
\mathrm{P} ({S^{1}, S^{2}}) \\
&=
\mathrm{P} (S^{1}) \mathrm{P} (S^{2} \,\vert\, S^{1}),
\end{align}
where $S^{1}$ is the event that the index chosen in the first iteration is successful and $S^{2}$ is the event that $k-1$ indices chosen in the second iteration are successful.

Our main result for the first iteration $\mathrm{P}\left(S^1\right)$ is as follows.
\begin{proposition}
The success probability of the first iteration in the CAST decoding satisfies
\begin{align}
\mathrm{P}(S^{1}) \geq \mathrm{P}\left( \lVert \tilde{\mathbf{v}} \rVert_{2} \leq \sqrt{\frac{\alpha m}{2k}}\left( 1 - \rho \right) \lVert \mathbf{h} \rVert_{\infty} \right),
\label{eq:P1_third}
\end{align}
where $\lVert \tilde{\mathbf{v}} \rVert_{2}$ is the $\ell_{2}$-norm of the noise $\tilde{\mathbf{v}}$, $\alpha$ is the desired SNR, $m$ is the number of measurements, $\rho = \sum\limits_{p=1}^{k}\frac{1}{m\left| \sin\frac{\pi (2i_{\omega_{p}}+1)}{2m} \right|}$
where $i_{\omega_p}$ ($\omega_p\in\Omega$) depends on the index chosen in the first iteration, $k$ is the number of nonzero elements, and $\lVert \mathbf{h} \rVert_{\infty}$ is the maximum channel gain.
\label{pro:pro001}
\end{proposition}
\begin{proof}
See Appendix A.
\end{proof}
\noindent Since the obtained lower bound of $\mathrm{P}\left(S^{1}\right)$ in \eqref{eq:P1_third} depends on two random variables $\lVert \tilde{\mathbf{v}} \rVert_{2}$ and $\lVert \mathbf{h} \rVert_{\infty}$, to compute the lower bound of $\mathrm{P}\left(S^{1}\right)$, we take the expectation of the conditional probability $\mathrm{P}\left(S^{1} \, \vert \, \lVert \mathbf{h} \rVert_{\infty} \right)$ with respect to the condition $\lVert \mathbf{h} \rVert_{\infty} = r$.
That is,
\begin{align}
\mathrm{P}\left(S^{1}\right)
& =
\label{eq:eq003}
\int_{0}^{\infty}\mathrm{P}\left( S^{1} \, \vert \, \lVert \mathbf{h} \rVert_{\infty} = r \right) f_{\lVert \mathbf{h} \rVert_{\infty}}(r)dr \\
& \geq
\int_{0}^{\infty}\mathrm{P}\left( \lVert \tilde{\mathbf{v}} \rVert_{2}^{2} \leq \frac{\alpha m}{2k}\left( 1 - \rho \right)^{2} r^2 \right)f_{\lVert \mathbf{h} \rVert_{\infty}}(r)dr
\end{align}
where $f_{\lVert \mathbf{h} \rVert_{\infty}}(r) = Nre^{-\frac{r^{2}}{2}}\left(1-e^{-\frac{r^{2}}{2}}\right)^{N-1}$\footnote{For analytic simplicity, we use the i.i.d Rayleigh fading channel model for $\mathbf{h}$~\cite{sesia2011lte}.}.
Since $\tilde{\mathbf{v}} \sim \mathcal{CN}(0,1)$, $\lVert \tilde{\mathbf{v}} \rVert_{2}^{2}$ is a Chi-squared random variable with $2m$ degree of freedom (DoF).
Using the cumulative distribution function (CDF) of $ \lVert \tilde{\mathbf{v}} \rVert_{2}^{2} $, we have
\begin{align}
\mathrm{P}\left( S^{1} \right)
\geq
\int_{0}^{\infty} \frac{\gamma\left( m, \frac{\alpha m}{2k}\left( 1 - \rho \right)^{2} r^2 \right)}{\Gamma(m)}Nre^{-\frac{r^{2}}{2}}\left(1-e^{-\frac{r^{2}}{2}}\right)^{N-1}dr,
\label{eq:prob_1st_final}
\end{align}
where $\Gamma(a)$ and $\gamma(a,b)$ are a complete gamma function and an incomplete gamma function, respectively.

We next present the success probability for the second iteration when the first iteration is successful.
\begin{proposition}
The success probability of the second iteration in the CAST decoding satisfies
\begin{align}
\mathrm{P}\left(S^{2} | S^{1}\right) \geq \left[1 - F\left(1|2,2,\zeta\right)\right]^{(k-1)(m-k)},
\label{eq:prob_2nd_final}
\end{align}
where $F(\cdot)$ is the CDF of the non-central F-distribution\footnote{The non-central $F$-distribution is described by the quotient $(X/n_{1})/(Y/n_{2})$ with the CDF given by
\begin{align}
F\left(x|n_{1},n_{2},\lambda\right)
&=
\sum_{r=0}^{\infty}\left(\frac{\left(\frac{1}{2}\lambda\right)^{j}}{j!}\exp\left(-\frac{\lambda}{2}\right)\right)I\left(\frac{n_{1}x}{n_{2}+n_{1}x}\mathrel{\stretchto{\mid}{4ex}}\frac{n_{1}}{2}+j,\frac{n_{2}}{2}\right)
\end{align}
where the numerator $X$ has a non-central chi-squared distribution with $n_{1}$ degrees of freedom and the denominator $Y$ has a central chi-squared distribution $n_{2}$ degrees of freedom.} and $\zeta$ is the noncentrality parameter depending on the channel realization.
\label{pro:pro002}
\end{proposition}
\begin{proof}
See Appendix B.
\end{proof}
\noindent From Proposition 1 and Proposition 2, we obtain the final result for $\mathrm{P}_{succ}$ as follows.
\begin{theorem}
The probability that the CAST-encoded packet is decoded successfully satisfies
\begin{align}
\mathrm{P}_{succ} \geq \left[ 1 - F\left(1|2,2,\zeta\right) \right]^{(k-1)(m-k)} \int_{0}^{\infty} \frac{\gamma\left( m, \frac{\alpha m}{2k}\left( 1 - \rho \right)^{2} r^2 \right)}{\Gamma(m)}Nre^{-\frac{r^{2}}{2}}\left(1-e^{-\frac{r^{2}}{2}}\right)^{N-1}dr.
\label{eq:prob_total}
\end{align}
\end{theorem}
\begin{proof}
Using \eqref{eq:prob_1st_final} and \eqref{eq:prob_2nd_final}, we obtain the desired result.
\end{proof}

\begin{figure}[t]
	\centering
\includegraphics[scale=0.7]{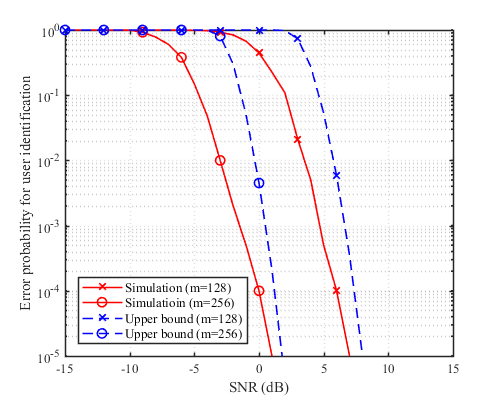}
\label{fig:lower_m_256}
    \caption{Empirical simulation results and upper bound of the error probability of support identification $(N=1024 \text{ and } \tau=2)$.}
\label{fig:8}
\end{figure}

\begin{figure}[t]
	\centering
	\includegraphics[width=0.55\columnwidth]{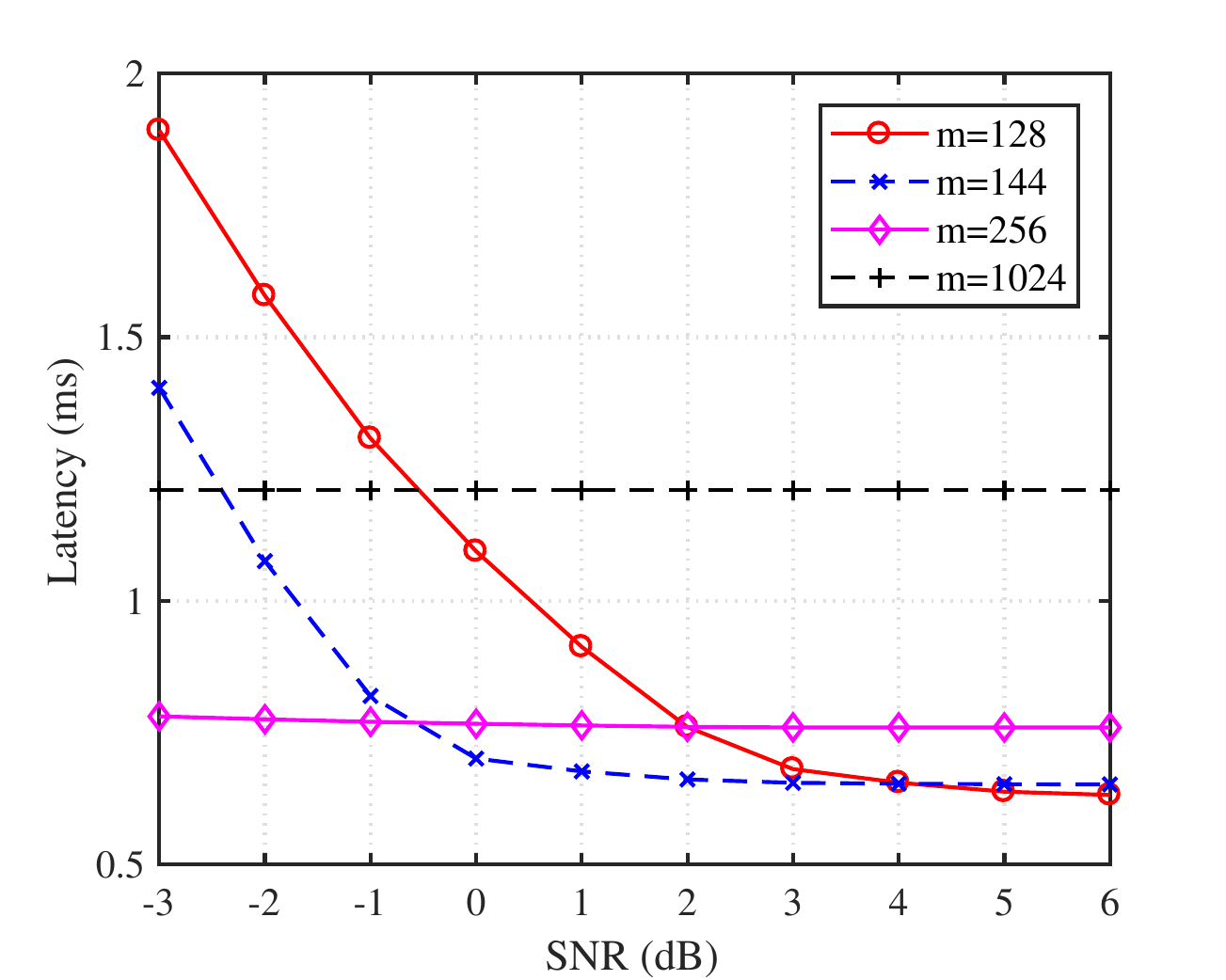}
    \caption{Average access latency for the uplink transmission as a function of SNR ($N=1024$, $k=9$, and $\tau=2$)}
    \label{sim:latency}
\end{figure}

In order to judge the effectiveness of the obtained lower bound in \eqref{eq:prob_total}, we plot the theoretical bound and empirical simulation results as a function of SNR for $m$ (see Fig.~\ref{fig:8}).
In this figure, we plot the error probability of user identification defined as $1-\mathrm{P}_{succ}$.
In our simulations, we compute the empirical averages to approximate the expectations with respect to $\rho$ and $\zeta$.
From these results, we observe that the obtained bound is tight, in particular for high SNR regime.
In the middle SNR regime, on the other hand, we observe some gap between the theoretical and empirical simulation results.
The gap is because the use of 1) an upper bound of column correlation and 2) the inequalities such as triangular inequality and Cauchy-Schwarz inequality (see Appendix. A).
From this figure, we also observe that the success probability increases sharply when the number of measurements $m$ increases.
For example, if $m$ is doubled from 128 to 256, we can achieve more than 5 dB gain in performance.

In many URLLC applications, latency and reliability are equally important and thus both should be considered in the system design and evaluation~\cite{revision}.
In the proposed scheme, when $m$ increases, the reliability will be improved but the latency will also increase due to the increase of the buffering latency and decoding latency.
In Fig. 8, we plot the mean access latency required to complete the CAST procedure for different values of $m$.
Note, if either the support identification or symbol detection is failed, the CAST procedure is repeated.
We observe that the proposed CAST scheme achieves the low access latency and also good decoding performance.
For example, when $m$ is reduced from 1024 to 256, the access latency is reduced by the factor of 35\%.
However, when $m$ is too small, the access latency is rather increased, in particular for low SNR regime, since in this case the CAST decoding can be failed and hence the entire process needs to be repeated.

\begin{figure}[t]
	\centering
	\subfloat{\includegraphics[width=0.55\columnwidth]{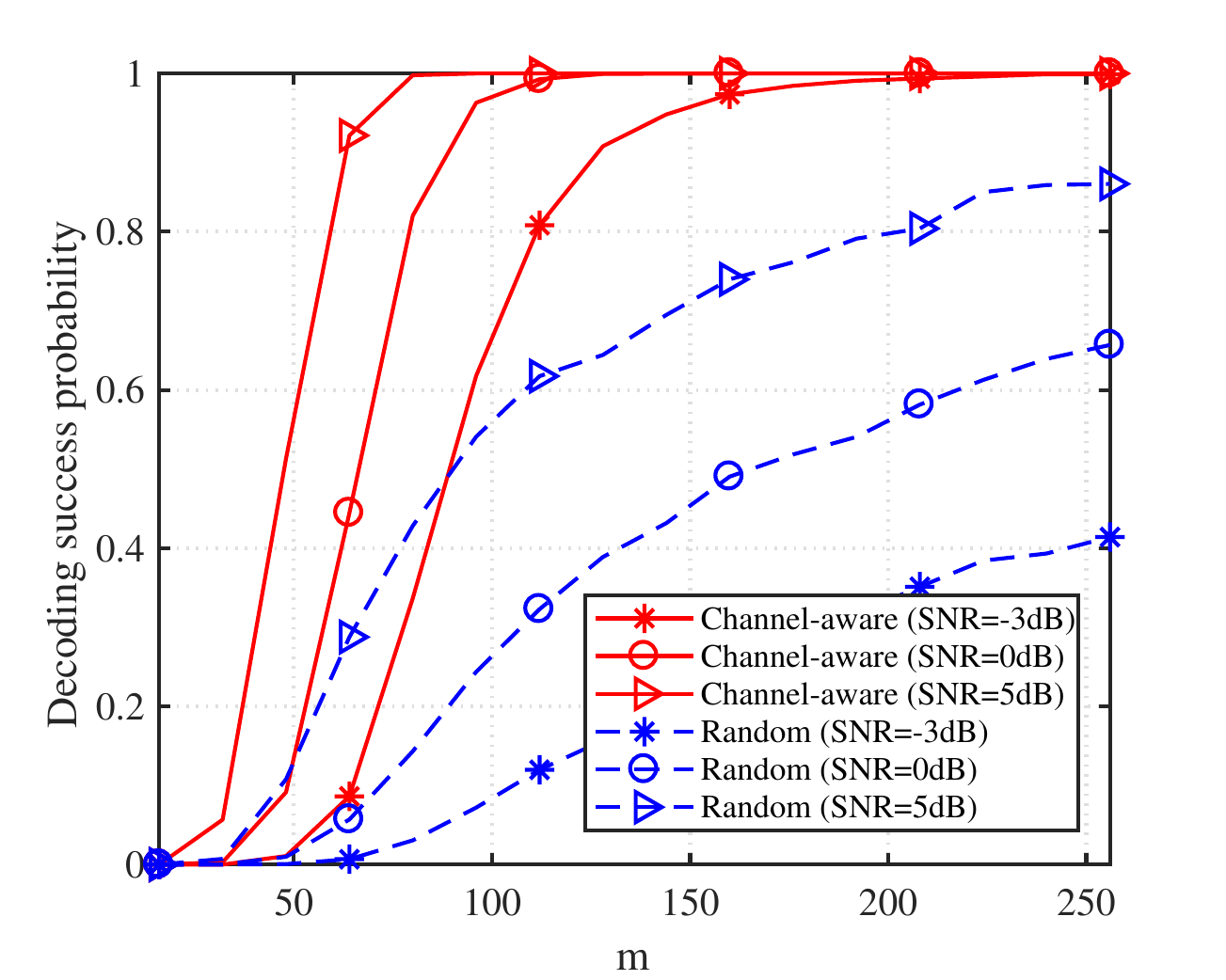}} \\
    \caption{Decoding success probability of the proposed CAST scheme as a funtion of $m$ under three different SNRs ($N=1024$, $k=6$, and $\tau=2$).}
    \label{sim:sim01}
\end{figure}

\section{Simulation Results}
In this section, we present the numerical results to evaluate the decoding performance and access latency of the proposed CAST.
In our simulations, we consider the OFDM-based TDD systems with $N=1024$ subcarriers.
As a channel model, we use the i.i.d Rayleigh fading channels.
For comparison, we use two different approaches in the support selection.
In the first approach, we choose the subcarriers uniformly at random among $N$ subcarriers.
In the second approach, we choose the support by the proposed selection rule (Algorithm 1).
In the decoding process, we use the proposed decoding algorithm (Algorithm 2) with $\tau$-close support identification ($\tau = 2$).
As performance metrics, we use the success probability of support identification, symbol error rate (SER), and also average access latency.
The access latency is defined as the sum of the waiting latency $T_{wait}$ and processing time $T_{proc}$ in~\eqref{eq:latency}.

\begin{figure}[t]
	\centering
	\subfloat[]{\includegraphics[width=0.5\columnwidth]{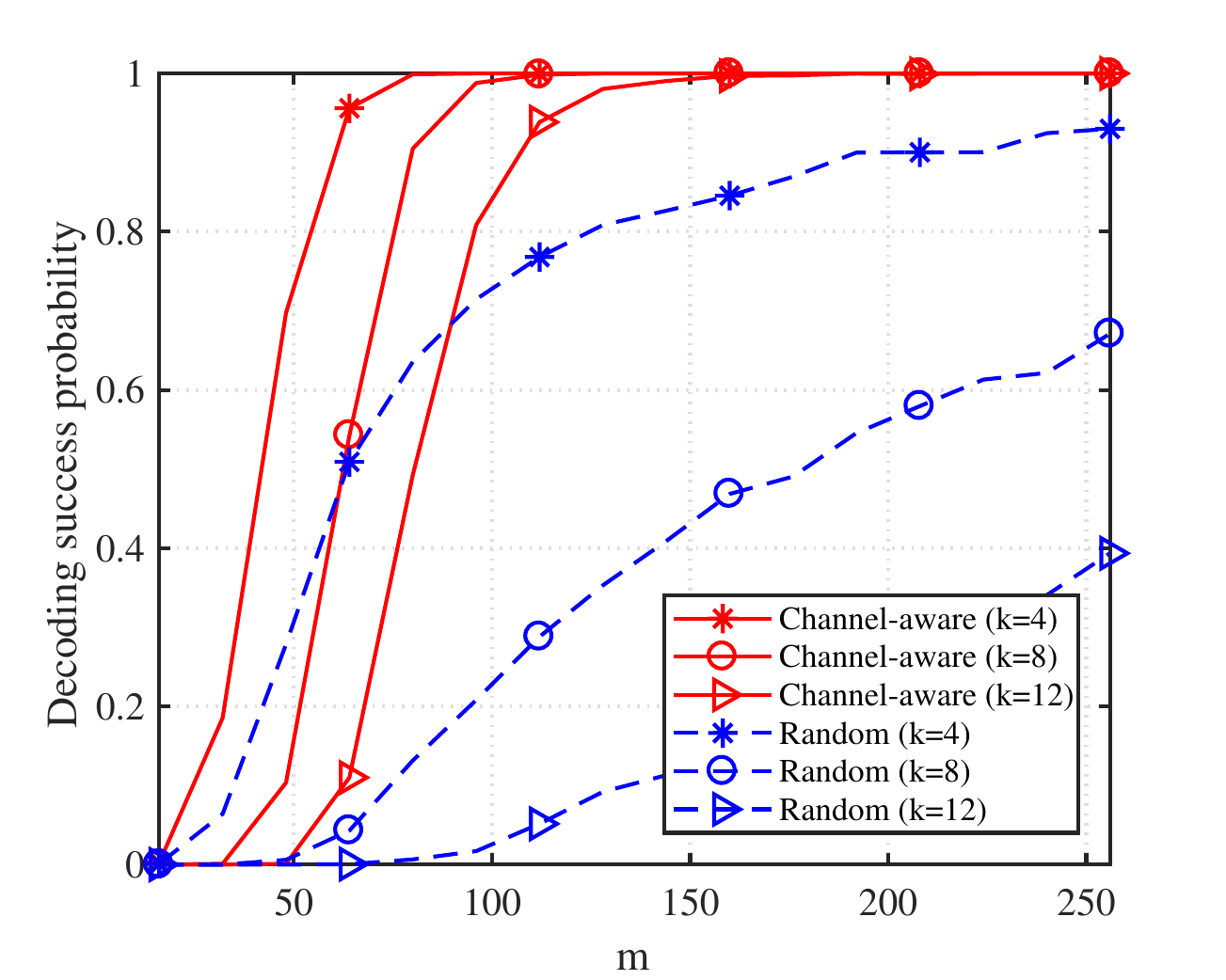}}
	\subfloat[]{\includegraphics[width=0.5\columnwidth]{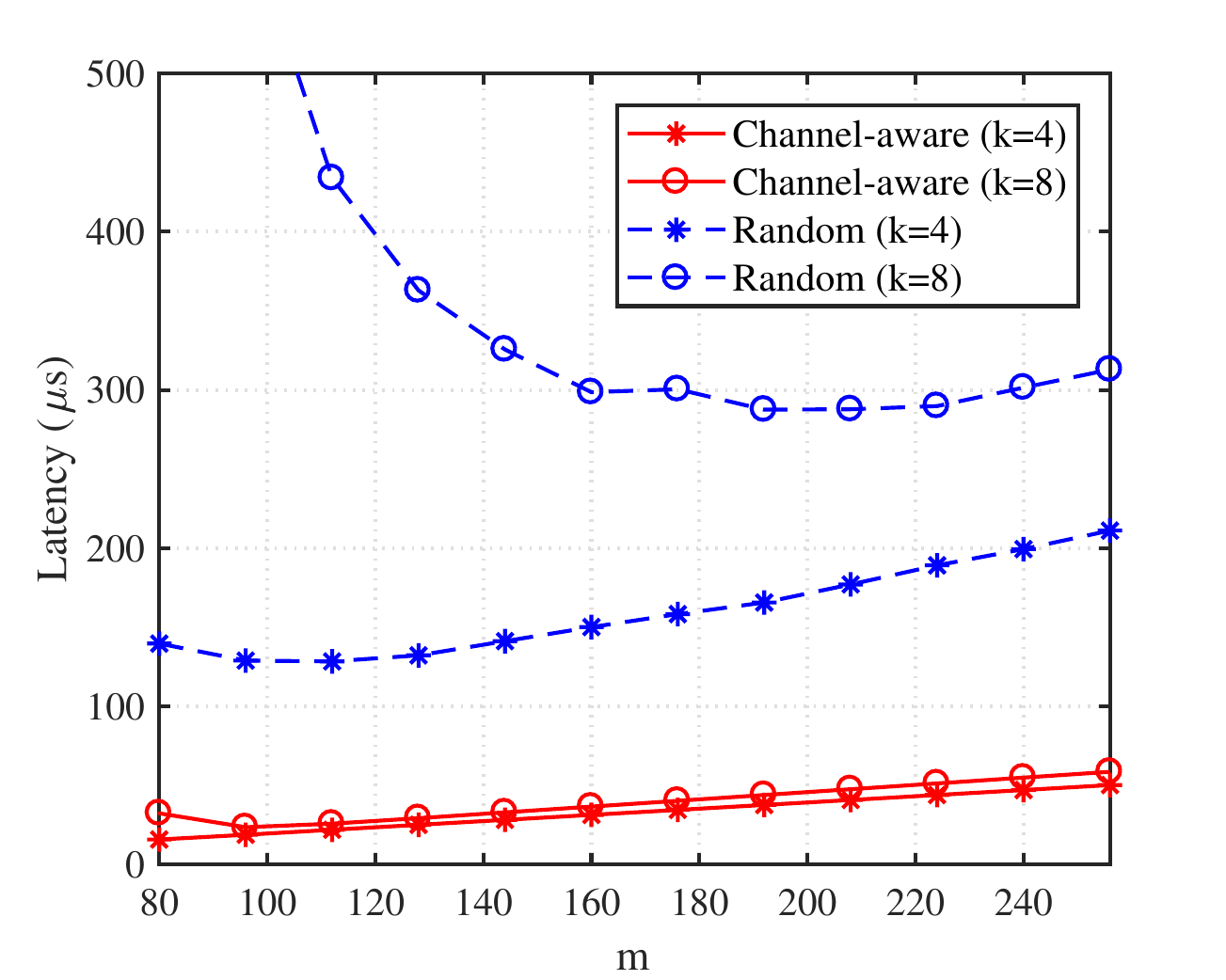}}
    \caption{CAST performances as a function of $m$ ($N=1024$, $\text{SNR}=3\text{dB}$, and $\tau=2$) : (a) Decoding success probability for different sparsity level ($k=4,8$ and 12). (b) Average latency for the CAST procedure.}
    \label{sim:k}
\end{figure}

In Fig.~\ref{sim:sim01}, we evaluate the success probability of the support identification as a function of $m$ for various SNRs (SNR = $-3 \text{dB}$, $0 \text{dB}$, and $5 \text{dB}$).
Simulation results demonstrate that the proposed CAST scheme achieves a significant reduction in the number of received samples.
When compared to the conventional signaling mechanism in which all received samples are needed to decode the grant information, CAST requires much smaller number of samples.
For example, CAST requires only 7.8\% ($m=80$ at 5 dB) of the received samples, which directly implies that the buffering latency $T_{buff}$ can be reduced by the factor of 92.2\% (see Section III.A).

\begin{figure}[t]
	\centering
	\includegraphics[width=0.55\columnwidth]{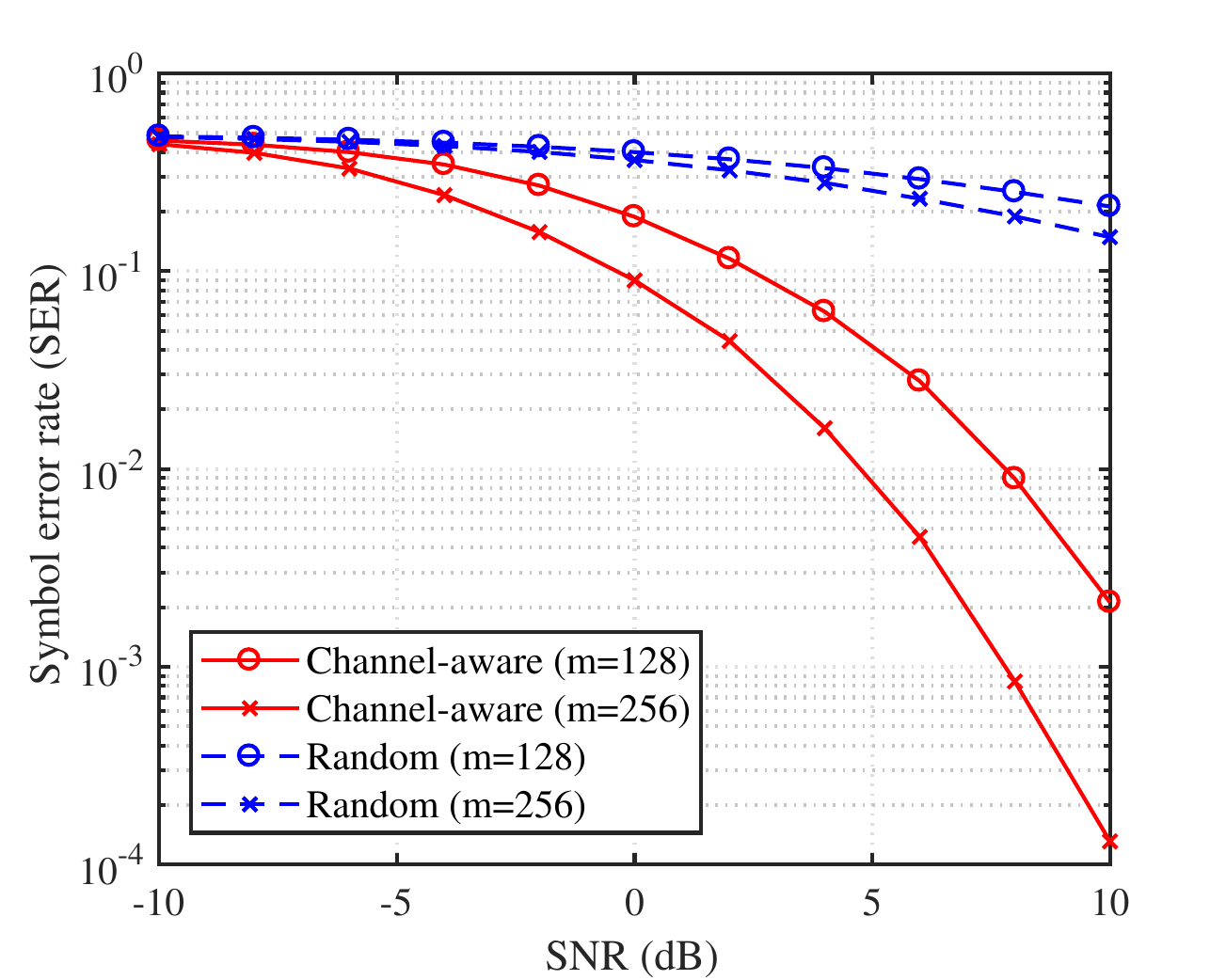}
    \caption{Symbol error rate for various number of received samples ($N=1024$, $k=10$, and $\tau=2$). In these simulations, the quadrature phase shift keying (QPSK) modulation is used.}
    \label{sim:ser}
\end{figure}


In Fig.~\ref{sim:k}(a), we evaluate the success probability of the support identification for various sparsity levels ($k=4, 8,$ and $12$).
We observe that only 10\% ($k=4$) and 15\% ($k=12$) of the received samples are needed to decode the grant information.
This behavior, however, cannot be achieved in the random support selection approach.
For instance, if $k$ increases from 4 to 12, the required number of samples to achieve 40\% success probability increases from 38 samples to 75 samples in the proposed support selection rule but that for the random support selection rule increases from 57 to 256.
Also, we investigate the average latency for performing the CAST process (see Fig.~\ref{sim:k}(b)).
These results clearly demonstrate that the proposed support selection rule (in Sec III.B) is very effective in reducing the latency.
For example, if $k$ increases from 4 to 8, the latency for the proposed support selection rule is about the same but that for the random support selection increases 2 times at $m=160$.

In Fig.~\ref{sim:ser}, we plot the SER performance of the proposed CAST scheme for two different number of measurements ($m=128$ and $256$).
We observe that the proposed selection rule outperforms the random selection rule by a large margin.
For example, when $m=256$, the proposed selection rule achieves $10^{-4}$ SER performance at $\text{SNR}=10\text{ dB}$ but the random selection approach cannot achieve this level of reliability even at high SNR.


In order to verify the robustness of CAST in real scenario, we test the block error rate (BLER) of CAST and the physical downlink control channel (PDCCH) in 4G when the channel is estimated.
As shown in Fig.~\ref{sim:imperfect_channel}, we observe that the CAST scheme outperforms the PDCCH, achieving more than 6 dB gain over the conventional PDCCH at $10^{-4}$ BLER point.
We also observe that the proposed scheme is insensitive to the channel estimation error.
For example, when BLER$=10^{-4}$, the gap between the perfect channel and imperfect channel for the proposed scheme is less than 1 dB but that for the PDCCH is around 3 dB.

\begin{figure}[t]
	\centering
	\includegraphics[width=0.55\columnwidth]{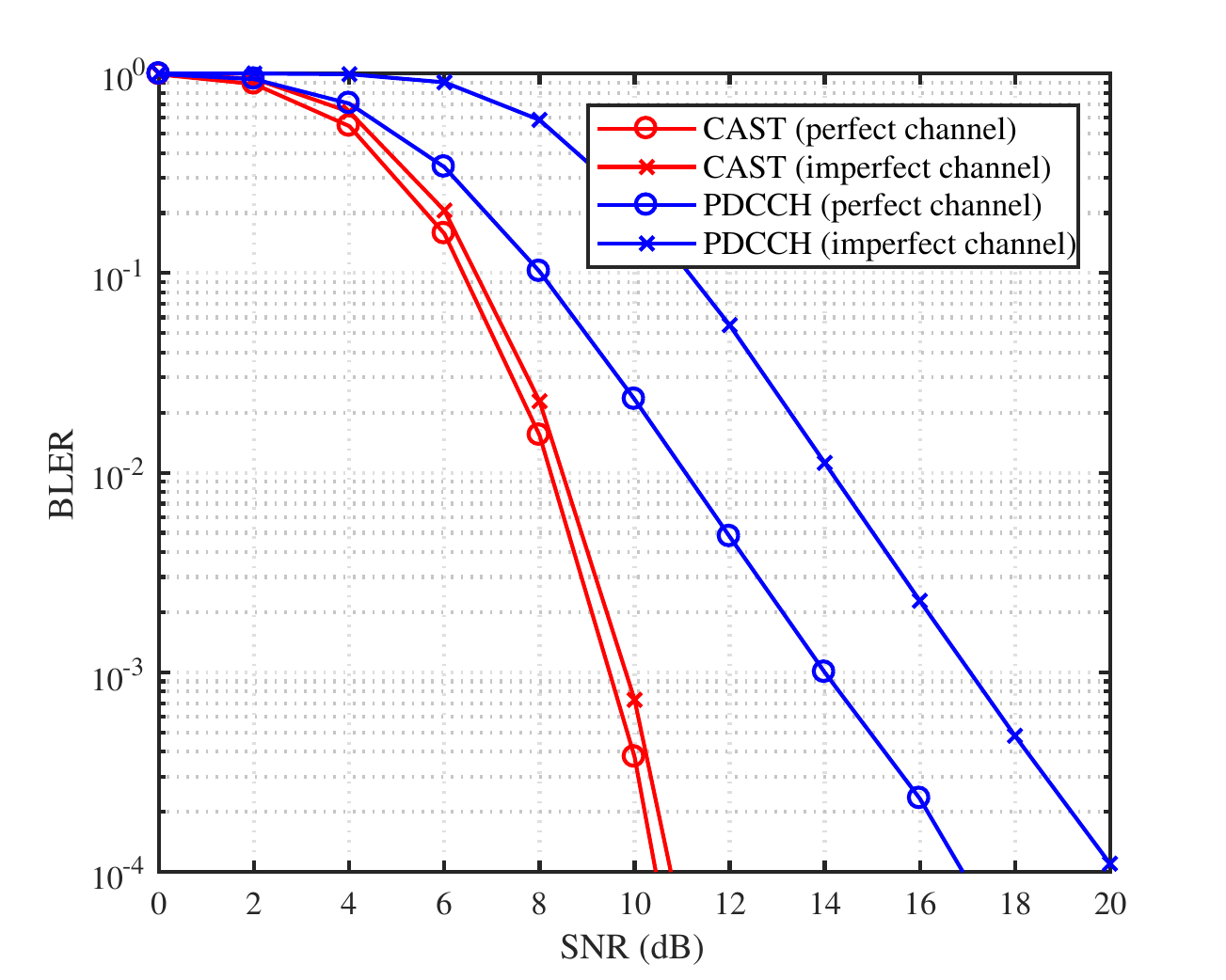}
    \caption{Block error rate of the CAST scheme and PDCCH using the perfect channel information and the estimated channel information.}
    \label{sim:imperfect_channel}
\end{figure}

Finally, we evaluate the access latency of CAST-based TDD system in Table.~\ref{Table_latency}.
In our simulation, we consider the LTE-TDD system (Rel. 13) and minislot-based NR TDD system (Rel. 15)\footnote{NR TDD system can flexibly schedule the UL data using the mini-slot (2,4 or 7 OFDM symbols) transmission. Using the mini-slot transmission, the switching period of NR TDD systems is shortened significantly and hence quick transmit direction change is possible. In this simulation, we use 2 OFDM symbols as a mini-slot.} as references.
The access latency in (1) can be expressed as $T_{up} = T_{wait} + T_{prop} + T_{proc} =$ $T_{wait} + T_{prop} + \left( \frac{m}{f_{s}} + T_{dec} \right)$ where $m$ is the number of received samples and $f_{s}$ is the sampling frequency.
When carrying out the mini-slot based access and CAST-based access, the base station changes the transmit direction into UL right after sending the grant signal and thus the mobile device can transmit the latency sensitive data without waiting for the periodic transmit direction change (i.e., $T_{wait} \approx 0$).
We use two TDD configurations with the different DL-UL ratio (9:1 and 8:2) and generate one URLLC packet in every two subframes.
In case of DL:UL=9:1 configuration, the access latency of the CAST-based TDD system (0.71 ms) is reduced by the factor of 87\% and 40\% over the LTE TDD system (5.56 ms) and NR TDD system (1.19 ms), respectively.
In a similar way, the access latency is also reduced by the factor of 82\% and 41\% for the DL:UL=8:2 configuration.
These results demonstrate that the CAST-based access is effective in the URLLC packet transmission.
In particular, when compared to the minislot-based NR TDD systems, we observe that the latency reduction obtained from CAST is non-negligible and meaningful.
This is because $T_{proc}$ is reduced substantially by using a small number of the received samples and simple decoding algorithm (see Section III.C).

\begin{table}[t]
\centering
\caption{Average latency under two different TDD configuration}
\label{Table_latency}
\begin{tabular}{|c|c|c|c|}
\hline
          & \textbf{Conventional LTE TDD} & \textbf{Minislot-based NR TDD} & \textbf{CAST-based TDD} \\ \hline
DL:UL=9:1 & 5.56ms                        & 1.19ms                       & 0.71ms                  \\ \hline
DL:UL=8:2 & 3.82ms                        & 1.16ms                       & 0.68ms                  \\ \hline
\end{tabular}
\end{table}

\section{Conclusion}
In this paper, we proposed the ultra low latency access scheme based on the CAST for URLLC.
Our work is motivated by the observation that waiting time to switch the transmit direction and processing time for the grant signal are quite large in TDD systems.
The key idea behind the proposed CAST scheme is to transform a URLLC grant information into the sparse symbol vector and to exploit the sparse recovery algorithm in decoding the sparse signal.
As long as the number of subcarriers is small enough and the measurements contain enough information to figure out the support and decode the grant information, accurate decoding of the CAST scheme can be guaranteed.
We demonstrated from the numerical evaluations that the proposed CAST scheme is very effective in TDD-based URLLC scenarios.
In this paper, we restricted our attention to the URLLC scenario but we believe that there are many interesting extensions worth investigating, such as the diversity support, machine learning-based CAST, and CAST for the FDD systems.

\begin{appendices}
\section{Proof of \eqref{eq:P1_third}}
Before we proceed to the main results, we provide the useful properties of the column correlation of $\mathbf{A}$ in \eqref{eq:column_correlation}.
\begin{lemma}
Recall that $f(\left\vert \omega_{p} - \omega_{q} \right\vert) = \left\vert\langle \mathbf{a}_{\omega_p}, \mathbf{a}_{\omega_q} \rangle \right\vert = \frac{1}{m}\left| \frac{\sin\frac{\pi m(\omega_p - \omega_q)}{N}}{\sin\frac{\pi (\omega_p - \omega_q)}{N}} \right|$ is the column correlation between $\mathbf{a}_{\omega_p}$ and $\mathbf{a}_{\omega_q}$ (see \eqref{eq:column_correlation}). Then the following statements hold true:
\begin{itemize}
\item[(i)] If $|\omega_p - \omega_q| = \frac{N}{m}, \frac{2N}{m}, \cdots, \frac{(m-1)N}{m}$, then $f(\left\vert \omega_{p} - \omega_{q} \right\vert) = 0$.
\item[(ii)] $f(\left\vert \omega_{p} - \omega_{q} \right\vert) \leq \frac{1}{m\left| \sin\frac{\pi (2i+1)}{2m} \right|}$ for some integer $i \geq 0$ satisfying $\max\left\{\frac{N}{2m},\frac{iN}{m}\right\} \leq |\omega_p - \omega_q| \leq \frac{(i+1)N}{m}$.
\end{itemize}
\label{lm:lm001}
\end{lemma}

\begin{proof}
In order to prove this proposition, we express the success probability $\mathrm{P}(S^{1})$ in terms of the column correlation of $\mathbf{A}$.
Specifically, let $\omega^{\ast} = \arg\underset{1 \leq \omega \leq N}{\max} \left\vert \langle \mathbf{a}_{\omega}, \tilde{\mathbf{y}} \rangle \right\vert$ be the index chosen in the first iteration.
Then, the first iteration would be successful if there exists only one $\omega \in \Omega = \{\omega_{1}, \cdots ,\omega_{k}\}$ satisfying $\vert\omega^{\ast} - \omega\vert < \tau$ (see Fig.~\ref{fig:tau_close_detection_model}).
Thus, we have
\begin{align}
\mathrm{P}(S^{1})
&=
\mathrm{P}( \left\vert \omega^{\ast} - \omega \right\vert < \tau, \ \text{for some} \ \omega\in\Omega ). \nonumber
\end{align}
Since the distance between two adjacent support elements is $\frac{N}{m}$ from Lemma~\ref{lm:lm001}(i), one can notice that $\tau$ should satisfy $\tau \leq \frac{N}{2m}$.
For analytic simplicity, we set $\tau = \frac{N}{2m}$ in our work.
Then we have
\begin{align}
\mathrm{P}(S^{1})
&=
\mathrm{P}\left( \left\vert \omega^{\ast} - \omega \right\vert < \frac{N}{2m}, \ \text{for some} \ \omega\in\Omega \right) \nonumber \\
&=
1 - \mathrm{P}\left( \left\vert \omega^{\ast} - \omega_{i} \right\vert \geq \frac{N}{2m}, \ \text{for all} \ \omega_{i}\in\Omega \right) \nonumber\\
&=
\label{eq:001}
1 - \mathrm{P}\left( \left\vert \omega^{\ast} - \omega_{1} \right\vert \geq \frac{N}{2m}, \cdots, \left\vert \omega^{\ast} - \omega_{k} \right\vert \geq \frac{N}{2m} \right).
\end{align}
First, we will find an upper bound of $\mathrm{P}\left( \left\vert \omega^{\ast} - \omega_{1} \right\vert \geq \frac{N}{2m}, \cdots, \left\vert \omega^{\ast} - \omega_{k} \right\vert \geq \frac{N}{2m} \right)$.
Let $\delta_{1} = \Big{[} \frac{N}{2m}, \frac{N}{m} \Big{]}$ and $\delta_{i} = \Big{(}\frac{(i-1)N}{m}, \frac{iN}{m}\Big{]}$ for $i=2,3,\cdots$, then $\Delta=\{\delta_{1}, \delta_{2},\cdots\}$ is a partition of the interval $\Big{[}\frac{N}{2m}, \infty\Big{)}$.
In this setting, it is clear that $\left\vert \omega^{\ast} - \omega_{i} \right\vert$ belongs to one interval in $\Delta$. 
In other words, $\left\vert \omega^{\ast} - \omega_{1} \right\vert \in \delta_{\omega_{1}}, \cdots, \left\vert \omega^{\ast} - \omega_{k} \right\vert \in \delta_{\omega_{k}}$ where $\delta_{\omega_{p}} = \Big{(}\max\left\{\frac{N}{2m},\frac{i_{\omega_{p}}N}{m}\right\}, \frac{(i_{\omega_{p}}+1)N}{m}\Big{]}$ for some $i_{\omega_{p}} \geq 0$ (see Fig.~\ref{fig:fig001}).
Therefore,

\begin{figure}[t]
\centering
\begin{tikzpicture}
\path (0,0) pic {myfunc={scale 1 xshift 0cm yshift 0cm}};
\end{tikzpicture}
\caption{If $|\omega^\ast-\omega_{p}| \geq \frac{N}{m}$, there exists a local maximum of $f(|\omega^\ast-\omega_{p}|)$ such that $f(|\omega^\ast-\omega_{p}|) \leq \frac{1}{m\left| \sin\frac{\pi (2i_{\omega_{p}}+1)}{2m} \right|}$. For example, if $\frac{N}{m} \leq |\omega^{\ast} - \omega_{1}| \leq \frac{2N}{m}$, then $f(|\omega^\ast-\omega_{1}|)\leq \frac{1}{m\left| \sin\frac{\pi (2i_{\omega_{1}}+1)}{2m} \right|}$. In a similar way, if $\frac{3N}{m} \leq |\omega^{\ast} - \omega_{2}| \leq \frac{4N}{m}$, then $f(|\omega^\ast-\omega_{2}|)\leq \frac{1}{m\left| \sin\frac{\pi (2i_{\omega_{2}}+1)}{2m} \right|}$.}
\label{fig:fig001}
\end{figure}
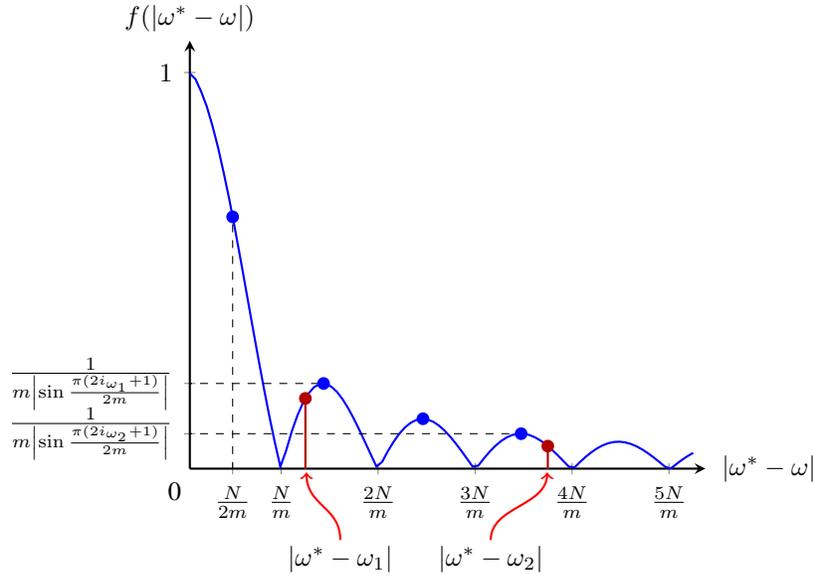

\begin{align}
&
\mathrm{P}\left( \left\vert \omega^{\ast} - \omega_{1} \right\vert \geq \frac{N}{2m}, \cdots, \left\vert \omega^{\ast} - \omega_{k} \right\vert \geq \frac{N}{2m} \right) \nonumber\\
& = 
\mathrm{P}\Big{(} \max\left\{\frac{N}{2m},\frac{i_{\omega_{1}}N}{m}\right\} \leq |\omega^{\ast} - \omega_1| \leq \frac{(i_{\omega_{1}}+1)N}{m}, \ \text{for some} \ i_{\omega_{1}}, \nonumber\\
& \qquad
\cdots, \max\left\{\frac{N}{2m},\frac{i_{\omega_{k}}N}{m}\right\} \leq |\omega^{\ast} - \omega_k| \leq \frac{(i_{\omega_{k}}+1)N}{m}, \ \text{for some} \ i_{\omega_{k}} \Big{)} \nonumber\\
&\stackrel{(a)}{\leq}
\mathrm{P}\left( f(\left\vert \omega^{\ast} - w_{1} \right\vert) \leq \frac{1}{m\left| \sin\frac{\pi (2i_{\omega_{1}}+1)}{2m} \right|}, \cdots, f(\left\vert \omega^{\ast} - w_{k} \right\vert) \leq \frac{1}{m\left| \sin\frac{\pi (2i_{\omega_{k}}+1)}{2m} \right|} \right) \nonumber \\
& =
\mathrm{P}\left( \left\vert\langle \mathbf{a}_{\omega^{\ast}}, \mathbf{a}_{w_{1}} \rangle \right\vert \leq \frac{1}{m\left| \sin\frac{\pi (2i_{\omega_{1}}+1)}{2m} \right|}, \cdots, \left\vert\langle \mathbf{a}_{\omega^{\ast}}, \mathbf{a}_{w_{k}} \rangle \right\vert \leq \frac{1}{m\left| \sin\frac{\pi (2i_{\omega_{k}}+1)}{2m} \right|} \right) \nonumber
\end{align}
\begin{align}
& \leq
\label{eq:002}
\mathrm{P}\left( \sum\limits_{\omega\in\Omega} \left\vert\langle \mathbf{a}_{\omega^{\ast}}, \mathbf{a}_w \rangle \right\vert \leq \sum\limits_{p=1}^{k}\frac{1}{m\left| \sin\frac{\pi (2i_{\omega_{p}}+1)}{2m} \right|} \right)
\end{align}
where (a) is from Lemma \ref{lm:lm001}(ii).
From \eqref{eq:001} and \eqref{eq:002}, we have
\begin{align}
\mathrm{P}(S^{1})
& \geq
\mathrm{P}\left( 
\sum\limits_{\omega\in\Omega} \left\vert\langle \mathbf{a}_{\omega^{\ast}}, \mathbf{a}_w \rangle \right\vert
\geq
\sum\limits_{p=1}^{k}\frac{1}{m\left| \sin\frac{\pi (2i_{\omega_{p}}+1)}{2m} \right|} 
\right) \nonumber\\
& = 
\mathrm{P}\left(
\sum\limits_{\omega\in\Omega} \left\vert\langle \mathbf{a}_{\omega^{\ast}}, \mathbf{a}_w \rangle \right\vert
\geq
\rho
\right) \nonumber
\end{align}
where $\rho = \sum\limits_{p=1}^{k}\frac{1}{m\left| \sin\frac{\pi (2i_{\omega_{p}}+1)}{2m} \right|}$.
Note that $\tilde{\mathbf{y}} = \sum\limits_{\omega\in\Omega}\mathbf{a}_{\omega}x_{\omega} + \tilde{\mathbf{v}} = \sum\limits_{\omega\in\Omega}\mathbf{a}_{\omega}h_{\omega}s_{\omega} + \tilde{\mathbf{v}} = \sum\limits_{\omega\in\Omega}\mathbf{a}_{\omega}\beta h_{\omega}\check{s}_{\omega} + \tilde{\mathbf{v}}$ where $\beta = \sqrt{\frac{2m\alpha}{k}}$ and $\check{s}_{\omega}$ is the normalized symbol.
Let $\left\vert h_{\omega_{l}} \right\vert = \max\limits_{\omega\in\Omega}\left\vert h_{\omega} \right\vert$, then we have

\begin{align}
\mathrm{P}(S^{1})
& \geq
\mathrm{P}\left( \beta\left\vert h_{\omega_{l}} \right\vert \sum\limits_{\omega\in\Omega} \left\vert\langle \mathbf{a}_{\omega^{\ast}}, \mathbf{a}_w \rangle \right\vert 
\geq
\beta\left\vert h_{\omega_{l}} \right\vert\rho \right) \\
& = 
\mathrm{P}\left(
\beta\left\vert h_{\omega_{l}} \right\vert\sum_{\omega\in\Omega} \left\vert \langle \mathbf{a}_{\omega^{\ast}}, \mathbf{a}_{\omega} \rangle\right\vert  + \left\vert\langle \mathbf{a}_{\omega^{\ast}}, \tilde{\mathbf{v}} \rangle \right\vert  
\geq
\beta \left\vert h_{\omega_{l}} \right\vert \rho + \left\vert\langle \mathbf{a}_{\omega^{\ast}}, \tilde{\mathbf{v}} \rangle \right\vert
\right)\\
& \geq
\mathrm{P}\left(
\beta\sum_{\omega\in\Omega} \left\vert \langle \mathbf{a}_{\omega^{\ast}}, \mathbf{a}_{\omega} \rangle\right\vert \left\vert h_{\omega} \right\vert + \left\vert\langle \mathbf{a}_{\omega^{\ast}}, \tilde{\mathbf{v}} \rangle \right\vert  
\geq
\beta \rho \left\vert h_{\omega_{l}} \right\vert + \left\vert\langle \mathbf{a}_{\omega^{\ast}}, \tilde{\mathbf{v}} \rangle \right\vert 
\right)
\phantom{\mathrm{P}(S^{1})} \\
& =
\label{eq:000}
\mathrm{P}\left( 
\sum_{\omega\in\Omega} \left\vert \langle \mathbf{a}_{\omega^{\ast}}, \mathbf{a}_{\omega} \rangle\right\vert \left\vert x_{\omega} \right\vert + \left\vert\langle \mathbf{a}_{\omega^{\ast}}, \tilde{\mathbf{v}} \rangle \right\vert
\geq
\beta \rho \left\vert h_{\omega_{l}} \right\vert  + \left\vert\langle \mathbf{a}_{\omega^{\ast}}, \tilde{\mathbf{v}} \rangle \right\vert 
\right)\\
& =
\label{eq:eq009}
\mathrm{P}\left(
\sum_{\omega\in\Omega} \left\vert \langle \mathbf{a}_{\omega^{\ast}}, \mathbf{a}_{\omega} \rangle x_{\omega} \right\vert  + \left\vert\langle \mathbf{a}_{\omega^{\ast}}, \tilde{\mathbf{v}} \rangle \right\vert 
\geq
\beta \rho \left\vert h_{\omega_{l}} \right\vert  + \left\vert\langle \mathbf{a}_{\omega^{\ast}}, \tilde{\mathbf{v}} \rangle \right\vert 
\right)\\
& \geq 
\label{eq:eq010}
\mathrm{P}\left(
\left\vert \sum_{\omega\in\Omega}  \langle \mathbf{a}_{\omega^{\ast}}, \mathbf{a}_{\omega} \rangle x_{\omega} + \langle \mathbf{a}_{\omega^{\ast}}, \tilde{\mathbf{v}} \rangle \right\vert
\geq
\beta \rho \left\vert h_{\omega_{l}} \right\vert  + \left\vert\langle \mathbf{a}_{\omega^{\ast}}, \tilde{\mathbf{v}} \rangle \right\vert 
\right)\\
& = 
\mathrm{P}\left(
\left\vert \langle \mathbf{a}_{\omega^{\ast}}, \sum_{\omega\in\Omega} \mathbf{a}_{\omega} x_{\omega} +  \tilde{\mathbf{v}} \rangle \right\vert
\geq
\beta \rho \left\vert h_{\omega_{l}} \right\vert  + \left\vert\langle \mathbf{a}_{\omega^{\ast}}, \tilde{\mathbf{v}} \rangle \right\vert 
\right)\\
& =
\label{eq:eq011}
\mathrm{P}\left( 
\left\vert \langle \mathbf{a}_{\omega^{\ast}}, \tilde{\mathbf{y}} \rangle \right\vert
\geq
\beta \rho \left\vert h_{\omega_{l}} \right\vert  + \left\vert\langle \mathbf{a}_{\omega^{\ast}}, \tilde{\mathbf{v}} \rangle \right\vert 
\right)
\end{align}
where \eqref{eq:000} is because $\left\vert x_{\omega} \right\vert = \beta\left\vert h_{\omega} \right\vert$ and \eqref{eq:eq010} is from the triangular inequality.

Since $\left\vert \langle \mathbf{a}_{\omega^{\ast}}, \tilde{\mathbf{y}} \rangle \right\vert \geq \left\vert \langle \mathbf{a}_{\omega_{l}}, \tilde{\mathbf{y}} \rangle \right\vert$, we further have
\begin{align}
\mathrm{P}(S^{1})
& \geq
\mathrm{P}\left(
\left\vert \langle \mathbf{a}_{\omega_{l}}, \tilde{\mathbf{y}} \rangle \right\vert
\geq
\beta \rho \left\vert h_{\omega_{l}} \right\vert  + \left\vert\langle \mathbf{a}_{\omega^{\ast}}, \tilde{\mathbf{v}} \rangle \right\vert 
\right) \\
& = 
\mathrm{P}\left(
\left\vert \langle \mathbf{a}_{\omega_{l}}, \sum_{\omega \in \Omega}\mathbf{a}_{\omega} x_{\omega} + \tilde{\mathbf{v}} \rangle \right\vert
\geq
\beta \rho \left\vert h_{\omega_{l}} \right\vert  + \left\vert\langle \mathbf{a}_{\omega^{\ast}}, \tilde{\mathbf{v}} \rangle \right\vert 
\right) \\
& = 
\label{eq:eq015}
\mathrm{P}\left(
\left\vert x_{\omega_{l}} + \langle \mathbf{a}_{\omega_{l}}, \tilde{\mathbf{v}} \rangle \right\vert
\geq
\beta \rho \left\vert h_{\omega_{l}} \right\vert  + \left\vert\langle \mathbf{a}_{\omega^{\ast}}, \tilde{\mathbf{v}} \rangle \right\vert 
\right) \\
& \geq 
\label{eq:eq016}
\mathrm{P}\left(
\left\vert x_{\omega_{l}} \right\vert - \left\vert \langle \mathbf{a}_{\omega_{l}}, \tilde{\mathbf{v}} \rangle \right\vert
\geq
\beta \rho \left\vert h_{\omega_{l}} \right\vert  + \left\vert\langle \mathbf{a}_{\omega^{\ast}}, \tilde{\mathbf{v}} \rangle \right\vert 
\right) \\
& = 
\mathrm{P}\left(
\beta \left\vert h_{\omega_{l}} \right\vert - \left\vert \langle \mathbf{a}_{\omega_{l}}, \tilde{\mathbf{v}} \rangle \right\vert
\geq
\beta \rho \left\vert h_{\omega_{l}} \right\vert  + \left\vert\langle \mathbf{a}_{\omega^{\ast}}, \tilde{\mathbf{v}} \rangle \right\vert
\right) \\
\phantom{\mathrm{P}\left( S^{1} \right)}
& = 
\mathrm{P}\left(
\beta \left\vert h_{\omega_{l}} \right\vert - \left\vert \langle \mathbf{a}_{\omega_{l}}, \tilde{\mathbf{v}} \rangle \right\vert - \left\vert\langle \mathbf{a}_{\omega^{\ast}}, \tilde{\mathbf{v}} \rangle \right\vert
\geq
\beta \rho \left\vert h_{\omega_{l}} \right\vert  
\right) \\
& \geq 
\label{eq:eq019}
\mathrm{P}\left(
\beta \left\vert h_{\omega_{l}} \right\vert - 2\| \tilde{\mathbf{v}} \|_2
\geq
\beta \rho \left\vert h_{\omega_{l}} \right\vert  
\right)
\end{align}
\begin{align}
& = 
\mathrm{P}\left( \sqrt{\frac{\alpha m}{2k}}\left( 1 - \rho \right) \left\vert h_{\omega_{l}} \right\vert \geq \lVert \tilde{\mathbf{v}} \rVert_{2} \right) \\
& = 
\label{eq:003}
\mathrm{P}\left( \sqrt{\frac{\alpha m}{2k}}\left( 1 - \rho \right) \lVert \mathbf{h} \rVert_{\infty} \geq \lVert \tilde{\mathbf{v}} \rVert_{2} \right),
\end{align}
where \eqref{eq:eq015} is because $\left\vert \langle \mathbf{a}_{\omega_{l}}, \mathbf{a}_{\omega_{l}} \rangle \right\vert = 1$ and $\left\vert \langle \mathbf{a}_{\omega_{l}}, \mathbf{a}_{\omega} \rangle \right\vert = 0$ for $\omega \in \Omega\setminus\{\omega_{l}\}$, \eqref{eq:eq016} is from the triangular inequality, \eqref{eq:eq019} is from the Cauchy-Schwarz inequality (i.e., $\left\vert \langle \mathbf{a}_{\omega}, \tilde{\mathbf{v}} \rangle \right\vert \leq \lVert \mathbf{a}_{\omega} \rVert_{2} \lVert \tilde{\mathbf{v}} \rVert_{2}=\lVert \tilde{\mathbf{v}} \rVert_{2}$), and \eqref{eq:003} is because $\lVert \mathbf{h} \rVert_{\infty} = \max\left\vert \mathbf{h} \right\vert = \left\vert h_{\omega_{l}} \right\vert $.
\end{proof}

\section{Proof of \eqref{eq:prob_2nd_final}}
Recall that in the second iteration, the proposed algorithm picks the remaining $k-1$ columns from the set of columns orthogonal to the column chosen in the first iteration\footnote{As mentioned, when $\omega^{\ast} \in \{\omega - \frac{N}{2m}, \cdots, \omega, \cdots, \omega + \frac{N}{2m}\}$ for some $\omega \in \Omega$, we can consider $\omega^{\ast}$ as $\omega$. This is because the mobile device already knows the true support using the channel reciprocity.}.
Let $\Psi$ be the index set of the orthogonal columns to $\mathbf{a}_{\omega^{\ast}}$.
Then, we have
\begin{align}
\mathrm{P}\left(S^{2} | S^{1}\right)
&=
\mathrm{P}\left( \underset{\omega_{i}\in\Omega\setminus\{\omega^{\ast}\}}{\min}\left\vert \langle \mathbf{a}_{\omega_{i}}, \tilde{\mathbf{y}} \rangle \right\vert^{2} > \underset{\omega_{j}\in\Psi\setminus\Omega}{\max}\left\vert \langle \mathbf{a}_{\omega_{j}}, \tilde{\mathbf{y}} \rangle \right\vert^{2} \right) \\
&=
\prod_{\omega_{i}\in\Omega\setminus\{\omega^{\ast}\}}\mathrm{P}\left( \left\vert \langle \mathbf{a}_{\omega_{i}}, \tilde{\mathbf{y}} \rangle \right\vert^{2} > \underset{\omega_{j}\in\Psi\setminus\Omega}{\max}\left\vert \langle \mathbf{a}_{\omega_{j}}, \tilde{\mathbf{y}} \rangle \right\vert^{2} \right) \\
& =
\prod_{\omega_{i}\in\Omega\setminus\{\omega^{\ast}\}}\prod_{\omega_{j}\in\Psi\setminus\Omega}\mathrm{P}\left( \left\vert \langle \mathbf{a}_{\omega_{i}}, \tilde{\mathbf{y}} \rangle \right\vert^{2} > \left\vert \langle \mathbf{a}_{\omega_{j}}, \tilde{\mathbf{y}} \rangle \right\vert^{2} \right).
\label{eq:prob_2nd_iter}
\end{align}
Let $\omega_{i^{\ast}} = \arg\min\limits_{\omega_{i}\in\Omega\setminus\{\omega^{\ast}\}} \left\vert \langle \mathbf{a}_{\omega_{i}}, \tilde{\mathbf{y}} \rangle \right\vert^{2}$ and $\omega_{j^{\ast}} = \arg\max\limits_{\omega_{j}\in\Psi\setminus\Omega} \left\vert \langle \mathbf{a}_{\omega_{j}}, \tilde{\mathbf{y}} \rangle \right\vert^{2}$, then all probability components in \eqref{eq:prob_2nd_iter} are lower bounded as $\mathrm{P}\left(\left\vert \langle \mathbf{a}_{\omega_{i^{\ast}}}, \tilde{\mathbf{y}} \rangle \right\vert^{2} > \left\vert \langle \mathbf{a}_{\omega_{j^{\ast}}}, \tilde{\mathbf{y}} \rangle \right\vert^{2}\right)$.
Hence,
\begin{align}
\mathrm{P}\left(S^{2} | S^{1}\right)
& \geq
\label{eq:0000}
\left[\mathrm{P}\left(\left\vert \langle \mathbf{a}_{\omega_{i^{\ast}}}, \tilde{\mathbf{y}} \rangle \right\vert^{2} > \left\vert \langle \mathbf{a}_{\omega_{j^{\ast}}}, \tilde{\mathbf{y}} \rangle \right\vert^{2}\right)\right]^{(k-1)(m-k)} \\
& =
\label{eq:2nd00}
\left[\mathrm{P}\left(\left\vert\frac{ \langle \mathbf{a}_{\omega_{i^{\ast}}}, \tilde{\mathbf{y}} \rangle }{ \langle \mathbf{a}_{\omega_{j^{\ast}}}, \tilde{\mathbf{y}} \rangle }\right\vert^{2} > 1 \right)\right]^{(k-1)(m-k)}
\end{align}
where \eqref{eq:0000} is because $\left\vert \Omega\setminus\{\omega^{\ast}\} \right\vert = k-1$ and $\left\vert \Psi\setminus\Omega \right\vert = m-k$.
One can easily show that $\left\vert \langle \mathbf{a}_{\omega_{i^{\ast}}}, \tilde{\mathbf{y}} \rangle \right\vert^{2}$ is a non-central Chi-squared random variable with 2 DoF and non-centrality parameter $\zeta = \beta \left\vert h_{\omega_{i^{\ast}}} \right\vert^{2}$, and $\left\vert \langle \mathbf{a}_{\omega_{j^{\ast}}}, \tilde{\mathbf{y}} \rangle \right\vert^{2}$ is a central Chi-squared random variable with 2 DoF.
Thus, $\left\vert \frac{\langle \mathbf{a}_{\omega_{^{\ast}i}}, \tilde{\mathbf{y}} \rangle}{\langle \mathbf{a}_{\omega_{j^{\ast}}}, \tilde{\mathbf{y}} \rangle} \right\vert^{2}$ is a non-central $F$-distribution whose CDF is
\begin{align}
\mathrm{P}\left( \left\vert \frac{\langle \mathbf{a}_{\omega_{i^{\ast}}}, \tilde{\mathbf{y}} \rangle}{\langle \mathbf{a}_{\omega_{j^{\ast}}}, \tilde{\mathbf{y}} \rangle} \right\vert^{2} < x \right)
&=
F\left(x|2,2,\zeta\right) \\
&=
\label{eq:2nd01}
\sum_{r=0}^{\infty}\left(\frac{\left(\frac{1}{2}\zeta\right)^{r}}{r!}\exp\left(-\frac{\zeta}{2}\right)\right)I\left(\frac{x}{1+x}\mathrel{\stretchto{\mid}{4ex}}1+r,1\right),
\end{align}
where $I(x|a,b)$ is the regularized incomplete beta function with parameters $a$ and $b$.
From \eqref{eq:2nd00} and \eqref{eq:2nd01}, we have
\begin{align}
\mathrm{P}\left(S^{2} | S^{1}\right) \geq \left[1 - F\left(1|2,2,\zeta\right)\right]^{(k-1)(m-k)},
\end{align}
which is the desired results.
\end{appendices}

\bibliographystyle{ieeetr}
\bibliography{Reference}

\end{document}